\documentclass[11pt]{article}
\usepackage{amssymb,amsbsy,amsfonts,amsmath,amsthm,xspace}
\usepackage{mathrsfs}
\usepackage{graphicx}
\usepackage{caption}
\usepackage{setspace}
\usepackage{enumitem}
\usepackage{subcaption}
\usepackage{multirow}
\usepackage{epstopdf}
\usepackage{epsfig}
\usepackage{url}
\usepackage[toc,page]{appendix}
\usepackage{float}
\usepackage{natbib}
\usepackage{color}
\usepackage{verbatim}
\usepackage{authblk}
\usepackage[normalem]{ulem}
\usepackage{hyperref}
\hypersetup{
    colorlinks=true,
    linkcolor=blue,
    filecolor=magenta,      
    urlcolor=cyan,
    citecolor=blue
}
 
\newcommand{\V}[1]{\ensuremath{\boldsymbol{#1}}\xspace}

\theoremstyle{definition}
\newtheorem{thm}{Theorem}

\newtheorem{algo}{Algorithm}
\newtheorem{prop}{Proposition}
\newtheorem{coro}{Corollary}
\newtheorem{lem}{Lemma}
\newtheorem{rem}{Remark}

    \newtheorem{ass}{Assumption}

\newcommand{\e}{\mathbb{E}}
\newcommand{\p}{\mathbb{P}}

\newcommand{\bR}{\mathbb{R}}

\newcommand{\mbone}{{\mathbf{1}}}

\newcommand{\diag}{\text{diag}}

\newcommand{\norm}[1]{\Vert{#1}\Vert}

\begin{document}

\def\text#1{\mbox{\rm #1}}
\def\overset#1#2{\stackrel{#1}{#2} }

\def\mb{\mathbf}
\def\mr{\mathrm}
\def\dsum{\displaystyle\sum}
\def\dint{\displaystyle\int}
\def\dfrac{\displaystyle\frac}

\renewcommand{\baselinestretch}{1}
%
%
\newcommand{\pkg}[1]{\textsf{#1}}

\title{Community models for networks observed through edge nominations}

\vspace{0.5cm}

\author{Tianxi Li\\
\vspace{-0.4cm}
Department of Statistics, University of Virginia\\
 Elizaveta Levina and Ji Zhu \\ Department of Statistics, University of Michigan}
\date{\today}
\maketitle

\begin{abstract}

Communities are a common and widely studied structure in networks, typically under the assumption that the network is fully and correctly observed.   In practice, network data are often collected by querying nodes about their connections.  In some settings, all edges of a sampled node will be recorded, and in others, a node may be asked to name its connections. These sampling mechanisms introduce noise and bias which can obscure the community structure and invalidate assumptions underlying standard community detection methods.  We propose a general model for a class of network sampling mechanisms based on recording edges via querying nodes, designed to improve community detection for network data collected in this fashion.   We model edge sampling probabilities as a function of both individual preferences and community parameters, and show community detection can be performed by spectral clustering under this general class of models.   We also propose, as a special case of the general framework, a parametric model for directed networks we call the nomination stochastic block model, which allows for meaningful parameter interpretations and can be fitted by the method of moments.  Both spectral clustering and the method of moments in this case are computationally efficient and come with theoretical guarantees of consistency.  We evaluate the proposed model in simulation studies on both unweighted and weighted networks and apply it to a faculty hiring dataset, discovering a meaningful hierarchy of communities among US business schools.

\end{abstract}

\section{Introduction}\label{sec:intro}
Networks have been widely used to describe relationships between individuals or interactions between units of complex systems in numerous fields, including biology, computer science, sociology and economics, and giving insights into many natural phenomena such as gene regulation, friendship formation, and eco-system evolution \citep{newman2010networks}.   Community detection, the task of clustering nodes into groups with relatively homogeneous connection patterns,  has been intensively studied since communities occur naturally in many real-world networks \citep{fortunato2010community, goldenberg2010survey}.  
Many statistical network models with communities have now been proposed, from the simple stochastic block model \citep{holland1983stochastic}  to more complex  extensions with  mixed membership \citep{airoldi2008mixed} or temporal evolution \citep{xu2013dynamic,matias2017statistical}.    Such models can provide a rigorous statistical framework and theoretical performance guarantees (see, for example, \citet{rohe2011spectral, zhao2012consistency}  for early work and  \cite{abbe2017community} for a recent review), as well as lead to improved algorithms, e.g. , \cite{joseph2016impact, gao2015achieving}.

A practical difficulty in many empirical studies of networks arises from imperfect data collection. We loosely use the term ``edge nomination'' for any situation where edge information is obtained through a data collection mechanism which may not record the entire networks. This may occur in observational studies where interactions or connections are recorded from observations of the experimenters \citep{zachary1977information,hass1991social,connor1992dolphin,gleiser2003community}, and interactions not observed  would be missing from the data.   This can also include traditional surveys, since many social networks are constructed by asking subjects  to name their friends or contacts \citep{michell1996peer, harris2009national}.    Sometimes these surveys limit how many friends one can name, and sometimes subjects may choose to name their friends selectively.   Another example is internet crawlers that only follow some of the paths \citep{clauset2015systematic,ji2016coauthorship}.    In all these situations, the missing edges may undermine the validity or efficiency of standard network analysis methods. 


One important property that often arises in the aforementioned settings is that which edges are missing may depend on the properties of the individual node reporting them, which automatically invalidates all missing completely at random assumptions.   
We will use the term \emph{nomination network} to refer to any situation where the missing edge pattern may depend on the node from which the edge information is collected.

Missing edges in networks can also be viewed as erroneous observations (a 0 instead of a 1).  
There has been a significant amount of work on denoising networks, which often considers both missing edges and falsely reported edges.  \cite{butts2003network} propose a Bayesian method to evaluate how reliable an observed network is. Following a similar set of model assumptions,  \cite{newman2018networka} propose a link prediction framework to recover underlying networks without specific structures.  \cite{newman2018networkb} extends this work to a general framework to estimate networks under non-informative observational errors. Under the framework of exponential random graph models,  \cite{handcock2010modeling} study ways to handle general ignorable missing mechanisms. Related link prediction problems are studied in \cite{zhao2017link,wu2018link}. However, \cite{zhao2017link} focus on the general model-free link prediction without specific structural assumptions and are not directly applicable to community detection problems. For networks with communities,  \cite{guimera2009missing} propose a Bayesian model and inference method to detect both missing and spurious edges. \cite{martin2016structural} take a similar modeling strategy but assume more flexible nonparametric error distributions.

All these models for noisy networks assume the missing mechanism is independent of any network structure such as communities. In some situations, this assumption is reasonable, for instance, for recording errors, but for a network resulting from a survey, this assumption is hard to justify. For example, in a high school survey of friendships, there may be individual differences in whether to name friends from their own``true” community.    Here the missing mechanism potentially depends on both the community structure and individual node characteristics, requiring different models from those used for network denoising.   Recently, \cite{le2017estimating} considered a scenario where the missing mechanism depends on the community labels of the node pairs in the context of jointly analyzing multiple networks sampled from the same probability model, but their method does not apply to a single network. 

In this paper, we introduce a general framework of modeling communities in networks collected from this type of edge nomination and collection procedure, where the observed relations suffer from missingness. The framework can be used for both unweighted and weighted networks.  We also propose a new directed network model we call the {\em nomination stochastic block model} (NSBM), a special case of the general framework which has interpretable model parameters.  Under this model, we propose computationally efficient algorithms based on spectral clustering and the method of moments for community detection and parameter estimation, respectively, and show statistical consistency for both communities and estimated parameters.

The rest of the paper is organized as follows. In Section~\ref{sec:model}, we introduce the general framework and the NSBM, along with algorithms for  community detection and parameter estimation. In Section~\ref{sec:theory}, we establish community detection consistency under the general framework and parameter estimation consistency under the NSBM.  Simulation studies are presented in Section~\ref{sec:sim},  and an application of NSBM to a business school faculty hiring network in U.S. universities in Section~\ref{sec:BSchool}.    Section~\ref{sec:con} concludes with a discussion.    Proofs and additional results of the business school data analysis are included in the Appendix.

 \section{Community models for networks with edge nominations}\label{sec:model}

While networks constructed by reported links are often converted to undirected by ignoring the source of the edge, any network where edge information is collected from the individual nodes is directed by nature. Therefore, we start with a brief review of the directed version of the standard SBM and then present the new nomination model. 

\subsection{A general nomination framework based on the directed stochastic block model}
The stochastic block model (SBM) \citep{holland1983stochastic} is one of the most widely used and well understood models for communities in a network. It has been shown to recover communities in various settings successfully, and can serve as a building block for more complicated models; see \cite{abbe2017community} for a thorough recent review.  

A  network of $n$ nodes can be represented by an $n\times n$ adjacency matrix $A$ such that each entry $A_{ij} = \mbone(i \to j)$ is 1 if there is an edge from node $i$ to node $j$ and 0 otherwise.     The standard SBM is defined for undirected networks, where $A_{ij} = A_{ji}$.  
The directed extension has been studied by \cite{rohe2016co}, which in our context, reduces to the following model: given $n$ nodes, a positive integer $K$ and a $K\times K$ matrix of probabilities $B$, let $c_i \in \{1, \dots, K\}$ be the community label of node $i$, and $\V{c}$ be the vector of  community labels.    Here we treat $\V{c}$ as fixed.   Let
$G_k = \{i: c_i = k\}$
be the set of nodes in community $k$ and $n_k  = |G_k|$. The entries of the adjacency matrix $A$ are then generated independently from the Bernoulli distribution with 
\begin{equation}\label{eq:DSBM}
P(A_{ij} = 1) = B_{c_ic_j} 
\end{equation}
The difference between the undirected and the directed models is that the directed model does not require $B$ to be symmetric.   

Errors in recording network edges are common and can arise in a variety of ways.   We focus on the situation when some connections are missing but all observed connections are true edges. This is different from the setting considered in \cite{zhao2017link, newman2018networka}, where falsely reported edges are also allowed, but it is a reasonable assumption in many applications \citep{zachary1977information,hass1991social,connor1992dolphin,gleiser2003community,clauset2015systematic,ji2016coauthorship}, as discussed in Section~\ref{sec:intro}.  In particular, this is exactly the setting for the network of hiring relationships analyzed in Section~\ref{sec:BSchool}.   Let $\tilde{A}$  be the adjacency matrix we observe, with potentially missing edges, where $\tilde{A}_{ij} = 1$ indicates node $i$ reported that there is an edge from it to node $j$.  The generating process for $\tilde{A}$ can be thought of as taking  the original network $A$ generated from model \eqref{eq:DSBM} and applying a binary nomination ``mask'' matrix $R \in \{0,1\}^{n\times n}$, so that the observed matrix is  given by 
$$\tilde{A} = A\circ R, $$
where $\circ$ is the element-wise Hadamard matrix product. Here $R_{ij} = 1$  indicates that node $i$ revealed its connection to node $j$. In this paper, we assume a nominated link is always a true link in $A$, while $\tilde{A}_{ij} = 0$ may result from either $A_{ij} = 0$ or $R_{ij} = 0$.

We base our model for $R_{ij}$ on the following two considerations. By nature of the edge nomination process, the probability of the edge $A_{ij}$ being reported by node $i$ should depend on node $i$.   It is natural to assume that it also depends on how close nodes $i$ and $j$ are, which can be expressed through $P_{ij}$.   We therefore propose the following general model for the observed network:  
\begin{align}\label{eq:generic-model-hierarchical}
A_{ij} &\sim \text{Bernoulli}(B_{c_ic_j}), \text{~~independently, ~} \notag \\
R_{ij} &\sim \text{Bernoulli}(f_i(B_{c_ic_j})), \text{~~independently,~} \\
\tilde{A}_{ij} &= A_{ij} \cdot R_{ij} \, , \notag
\end{align}
where $f_i: [0,1] \to [0,1]$ is the {\em nomination function} of node $i$. This general form includes some of the previously studied settings.   For example,  when $f_i \equiv \rho_i$ for all $i \in [n]$, every node randomly nominates each of its links with a fixed probability $\rho_i$, the setting studied in \cite{butts2003network}. If for some nodes $f_i \equiv 0$ (all links are missing) and for others $f_i \equiv 1$ (all links from node $i$ are reported), we obtain the egocentric sampling setting  \citep{handcock2010modeling,krivitsky2017inference, wu2018link}.

In most situations, we are interested in learning about properties of the network expressed in $\V{c}, B$, or $f_i$'s rather than predicting the latent status $R_{ij}$.   We can integrate out $R_{ij}$ and write the distribution of $\tilde{A}$ directly as 
\begin{align}\label{eq:generic-model}
\p(\tilde{A}_{ij}=1) = \tilde P_{ij} = B_{c_ic_j}f_i(Bc_ic_j) = F_{i}(B_{c_ic_j})
\end{align}
where $F_i(x) = xf_i(x)$.   The general model  defined by \eqref{eq:generic-model} can be specialized to many different forms by specifying $f_i$.   One advantage of model  \eqref{eq:generic-model} is explicitly incorporating an informative missing mechanism through its dependence on the strength of the connection between the nodes.

\subsection{Community detection under the general edge nomination model}

One advantage of writing out a general model like \eqref{eq:generic-model} is the possibility of developing a general algorithm for solving problems of this type.  For the standard SBM and its degree corrected version \citep{karrer2011stochastic}, spectral clustering algorithms are among the most popular methods for estimating community labels \citep{rohe2011spectral, lei2014consistency, jin2015fast}, due to their many advantages: easy implementation, computational efficiency, and excellent theoretical properties.    Generally speaking, spectral clustering is used for community detection problems because under many models the community information can be recovered from the eigenvectors of the adjacency matrix or its Laplacian.   We next investigate whether this principle applies to our general model.   

Intuitively, since each node uses an individual function $F_i$ to nominate links which is applied to its entire row, we would expect the individual nomination preferences confounded in the row space of $\tilde A$.   However, the {\em column} space of $\tilde A$ should still reflect communities, since each node $i$ applies the same function to all entries of the column $j$.   This intuition is formalized in the following proposition about the singular value decomposition (SVD) of $\tilde{P}$.

\begin{prop}\label{lem:eigen}
Let $\tilde{P} = \tilde{U}\tilde{D}\tilde{V}^T$ be the SVD of $\tilde{P}$. There exists a matrix $X \in \bR^{K\times K}$ such that 
\begin{equation}\label{eq:eigen-structure}
\tilde{V} = ZX
\end{equation}
where $Z$ is the $n\times K$ community membership matrix, defined by $Z_{ik} = 1(c_i = k)$. In addition, $\norm{X_{k\cdot}-X_{l\cdot}}_2 = \sqrt{n_k^{-1}+n_{l}^{-1}}$ for any $1\le k \ne l \le K$.
\end{prop}

Proposition \ref{lem:eigen} suggests the right singular vectors of $\tilde{A}$ can be used to recover communities, as long as $\tilde{A}$ concentrates around $\tilde{P}$.  
We formalize this approach in the following algorithm, which we call  ``Right singular vectors Spectral Clustering" (Right SC).

\begin{algo}[Right SC]\label{algo:RightSC}
Given an adjacency matrix $\tilde{A}$ and the number of communities $K$:
\begin{enumerate}
\item Compute the rank $K$ truncated SVD $\tilde{A}$,  given by $\tilde A = \widehat{U}\widehat{D}\widehat{V}^T$.
\item Run  the $K$-means clustering algorithm on rows of $\widehat{V}$ to assign each node to a community.  
\end{enumerate}
\end{algo}

While optimizing the $K$-means loss is NP-hard, there are many efficient algorithms that find approximate solutions.    For obtaining theoretical guarantees, we will assume a version of the $K$-means algorithm that finds a value of the objective function that is at most $(1+\epsilon)$ of the global minimum;  this can be found efficiently for a small positive constant $\epsilon$ \citep{kumar2004simple}.

We will show that the Right SC algorithm mis-clusters a vanishing proportion of nodes  (see Section~\ref{sec:theory}). However, even though in practice we often observed exact recovery, it is known to be hard to prove an exact recovery guarantee for $K$-means based algorithms, except when there is some special structure \citep{abbe2017entrywise, li2018hierarchical, lei2020consistency}.

Instead, we introduce another spectral method, Spectral Minimum Spanning Tree Clustering (Right SMST), which also uses the right singular vectors.   The clusters are obtained by cutting the minimum spanning tree between embedded nodes; an algorithm studied in \cite{vu2018simple} and \cite{lei2017generic}. This algorithm is much easier to analyze than K-means, and we show in Section~\ref{sec:theory} that it can achieve the exact recovery of community labels for all nodes. However, in practice, the Right SC is always faster, and more importantly, much better on sparse networks. Therefore, the SMST algorithm is primarily of theoretical interest.  

\begin{algo}[Right SMST]\label{algo:RightSMST}
Given an adjacency matrix $\tilde{A}$ and the number of communities $K$:
\begin{enumerate}
\item Compute the rank-$K$ truncated singular value decomposition $\tilde{A} = \widehat{U}\widehat{D}\widehat{V}^T$.
\item Run minimum spanning tree algorithm of \cite{vu2018simple} on $\widehat{V}$:
\begin{enumerate}
\item Construct the undirected distance graph between $n$ nodes based on the distance matrix, where the edge weight between $i$ and $j$ is the distance  between $\widehat{V}_{i\cdot}$ and $\widehat{V}_{j\cdot}$, the $i$-th and $j$-th rows of the matrix $\widehat{V}$.
\item Find the minimum spanning tree of the distance graph. 
\item Remove the $K-1$ edges with the highest weights  from the minimum spanning tree .
\item Return the resulting connected components as clusters.
\end{enumerate}  
\end{enumerate}
\end{algo}

Intuitively, it is not hard to see why Algorithm~\ref{algo:RightSMST} may be inferior to Algorithm~\ref{algo:RightSC} in practice.  Algorithm~\ref{algo:RightSMST} is designed with the expectation that the between-cluster distances are always larger than within-cluster distances. This works when the signal is strong enough, but with weaker signal the minimum between-cluster distance and the maximum within-cluster distance can be unstable.  In contrast, $K$-means looks at the average behavior of observations within the same cluster and thus can be a lot more stable.  

\subsection{The nomination stochastic block model (NSBM)}\label{secsec:NSBM}

The general model \eqref{eq:generic-model} allows for a common algorithm for community detection. However, in many situations, in addition to community labels, one may also be interested in learning the nomination pattern $F_i$.  Making the so far unspecified functions $F_i$ both estimable and interpretable requires further modeling.   Next, we introduce a specific nomination model under the framework of \eqref{eq:generic-model}, which we believe achieves a good balance between generality and interpretability.

In addition to the previously defined $\V{c}$ and $B$, we introduce two new node-specific parameters, given by $n$-dimensional vectors $\V{\lambda} = (\lambda_i)$ and $\V{\theta} = (\theta_i)$.   The proposed nomination stochastic block model (NSBM) assumes
\begin{equation}\label{eq:NSBM-F}
f_i(P_{ij}) = \theta_i P_{ij}^{\lambda_i-1}, i \in [n]. 
\end{equation}
This model includes the special case of $F_i = \rho_i$.   More importantly, the parameters $\V{\lambda}$ and $\V{\theta}$ are easily interpretable. We can think of the parameter $\theta_i$ as measuring the overall propensity of node $i$ to nominate links, and of $\lambda_i$ as a measure of their preference for nominating links from their own or closely connected communities; both these factors may affect data collection in surveys \citep{harris2009national}. For example, suppose that $B_{kk} > B_{kj}, k\ne j$ so that the SBM is assortative. In this case,  $\lambda_i > 1$ indicates that the node $i$ tends to nominate connections from its own community while $\lambda_i <1$ indicates a preference for nominating connections from a different community. 

We can again write out the marginal distribution of $\tilde{A}$, given by 
\begin{equation}\label{eq:NSBM}
\p(\tilde{A}_{ij} = 1) = \theta_i B_{c_ic_j}^{\lambda_i}.
\end{equation}

As with any model involving products of parameters, we need additional constraints for identifiablity, as well as for ensuring the probability is always between 0 and 1.
We require that $\tilde{P}$ has no rows consisting entirely of zeros, and thus we require $\theta_i>0$ for all $i$ and that each row of $B$ contains at least one positive entry.     In addition, if $B_{kl}=0$ for all $l\ne k$, then community $k$ will not send edges to other communities and it will be impossible to identify $\lambda_i$'s for nodes in community $k$.    On the other hand, if $B_{kl} = B_{kk}$ for all $l$, then community $k$ is not identifiable.   We also need scaling constraints on  $B$ and $\V{\lambda}$ to avoid invariance multiplying by a constant.    Putting all these together leads to the following identifiability conditions.

\begin{prop}\label{prop:identify}
Model \eqref{eq:NSBM} is identifiable if the following conditions hold: 
\begin{enumerate}
\item $B_{kk}=1$ for all $k = 1, \dots, K$.  
\item For each $k$, there exists at least one $l \ne k$ such that $B_{kl} \ne B_{kk}$ and $B_{kl} \ne 0$. 
\item $\theta_i>0$ for all $i = 1, \dots, n$.
\item $\frac{1}{n_k} \sum_{i \in G_k}\lambda_i =1$ for all $k = 1, \dots, K$, where $n_k = |G_k|$.
\end{enumerate}
\end{prop}
Compared with the general model \eqref{eq:generic-model}, the NSBM offers the possibility of fitting an interpretable nomination mechanism model and learning something about each node's preferences.  We will illustrate this in our data analysis in Section \ref{sec:BSchool}, where the estimated nomination parameters offer a lot of insight into the underlying process.  The price we pay for interpretability, as usual, is less flexibility, since we now have more explicit model assumptions. For example, the requirement $\theta_i>0$ excludes the egocentric sampling mechanism, whereas the general model includes it.

\subsection{Parameter estimation under the NSBM}\label{secsec:NSBM-parameter-est}
Given community labels $\V{c}$, the other parameters in model \eqref{eq:NSBM} can estimated by the method of moments under the identifiability constraints of Proposition~\ref{prop:identify}.   Specifically, if $B_{kl} >0$, then for any arbitrary $i \in G_k$ and $j\in G_l$, we have
\begin{equation}
\label{eq:MomentRelation}
\log(\tilde{P}_{ij}) = \mu_{il} = \log(\theta_i) + \lambda_i\log(B_{kl}).
\end{equation}
Combining the conditions in Proposition~\ref{prop:identify} and \eqref{eq:MomentRelation}, we obtain the following identities:
\begin{eqnarray}
\theta_i & = & \tilde{P}_{ij} \text{ for any }  j \in G_{k} \, , \label{eq:id1}  \\
B_{kl} & = & \exp(- \frac{1}{n_k}\sum_{i \in G_k}(\mu_{ik} -\mu_{il})) \, , \\
\lambda_i & = & \frac{\mu_{ik}-\mu_{il}}{\sum_{j\in G_k}(\mu_{jk}-\mu_{jl})/n_k} \, , \text{ if } B_{kk} \ne B_{kl} \, .  \label{eq:id3}
\end{eqnarray}

Therefore, we can use the method of moments to estimate $\mu_{il}$ by  
\begin{equation}\label{eq:MomentMatching}
  \exp(\hat \mu_{il}) = \frac{1}{n_{l}}  \e \sum_{j \in G_{l}}\tilde{A}_{ij}
\end{equation}
and plug it \eqref{eq:id1}-\eqref{eq:id3} to estimate other parameters, with some modifications to handle boundary cases.   This is summarized in the following algorithm. 

   \begin{algo}[Parameter estimation for the NSBM by the method of moments]\label{algo:PE}
Given the adjacency matrix $\tilde{A}$ and community labels $\V{c}$, for $k = 1,2, \cdots, K$ (obtained by, for example, right SC):   

\begin{enumerate}
\item Set $T_{il} = \frac{\sum_{j \in G_l\tilde{A}_{ij}}}{n_{l}}$ for each $i \in G_k$ and $1 \le l \le K$. 
\item Estimate $\theta_i$ for each $i \in G_k$ by
\begin{equation}\label{eq:theta-estimate}
\hat{\theta}_i = T_{ik}\vee \frac{1}{n_k} \ .
\end{equation}

\item Find the set $\Psi_{k} = \{l: 1\le l \le K,  T_{il} > 0~~ \forall i \in G_k\}$. Set $\hat{B}_{kl} = 0$ for each $l \notin \Psi_k$.
\item \begin{enumerate}
\item Define $Y_{il} = \log(T_{il}\vee\frac{1}{n_l})$ for each $i \in G_k$, where the $\frac{1}{n_l}$ is used to avoid overflow for the pathological case of $T_{il}=0$ for some $i \in G_k$.
\item Estimate $\lambda_i$ for each $i\in G_k$ by
\begin{equation}\label{eq:lambda-esitmate}
\hat{\lambda}_i = \frac{\sum_{l \in \Psi_k/\{k\}}(Y_{ik}-Y_{il})}{\sum_{l \in \Psi_k/\{k\}}\sum_{j\in G_k}(Y_{jk}-Y_{jr})/n_k}
\end{equation}
\item Estimate $B_{kl}$ for each for each $l \in \Psi_k/\{k\}$ by 
\begin{equation}\label{eq:B-estimate}
\hat{B}_{kl} = \exp(- \frac{1}{n_k}\sum_{i \in G_k}(Y_{ik} -Y_{il})).
\end{equation}
\end{enumerate}
\end{enumerate}  

\end{algo}

\begin{rem} In the setting of unweighted networks, the method of moments estimators coincide with the MLE, as $T_{il}$ is the MLE of $\exp(\mu_{il})$. However, in the general weighted setting to be introduced in Section~\ref{sec:extent}, the MLE may be hard to obtain. In contrast, the method of moments remains a computationally feasible option as it only requires conditions on first-order moments. 
\end{rem}

\subsection{Extension to weighted networks}\label{sec:extent}


Networks with weighted edges are frequently encountered in practice, and even though methods for binary networks can be applied to weighted networked after thresholding edge weights, this approach results in substantial loss of information and can be sensitive to the choice of threshold.  Fortunately, we can extend the NSBM directly to weighted networks, without applying a thresholding step.


Given community labels $\V{c}$, assume each edge weight $\tilde{A}_{ij}$ is independently generated from a probability distribution satisfying 
\begin{equation}\label{eq:weightedPowerModel}
\e_{\pi}\tilde{A}_{ij} = \theta_iB_{c_ic_j}^{\lambda_i}.
\end{equation}

The specific choice of distribution class will depend on the problem at hand. For instance, the Poisson distribution has often been used to model network edge weights and is a popular choice for non-negative integer weights \citep{karrer2011stochastic}. In this case,  we can interpret the true $A_{ij}$ as the number of interactions directed  from node $i$ to node $j$ and model it as generated from $\text{Poisson}(B_{c_ic_j})$.   Then $\tilde{A}_{ij}$ can be interpreted as the number of interactions chosen randomly by node $i$ from $\text{Binomial}(A_{ij}, \theta_iB_{c_ic_j}^{\lambda_i-1})$ to report as their relationship with node $j$.  Again, we are assuming that only true interactions are reported, just like in the unweighted case, so that $\tilde A_{ij} \le A_{ij}$.

Since the model is specified through the expectation, we can still apply the right spectral clustering and method of moments algorithms, and similar theoretical guarantees can be obtained as long as the generating distributions of the edge weights are  not heavy-tailed. Note that we have only specified the mean structure of $\tilde{A}$ in  \eqref{eq:weightedPowerModel}, so if the edge distribution depends on other parameters, additional constraints and modifications may be necessary.

\section{Theoretical properties}\label{sec:theory}
  
Here we investigate asymptotic properties of community detection under the general model \eqref{eq:generic-model} and parameter estimation under the NSBM. We always assume the $B$ matrix is full-rank, and the number of communities $K$ is known and fixed. In practice,  $K$ can be estimated by many data-driven methods such as the edge cross-validation of \cite{li2016network} by taking advantage of the property that the rank of $\tilde{P}$ equals the number of communities.  
    
    \subsection{Consistency of community detection}\label{secsec:CommunityTheory}
 
We first introduce an additional assumption we need for theoretical developments, which is that none of the communities vanish relatively to the size of others as $n$ grows.

    \begin{ass}\label{A2} 
  Assume that $n_{\min}:= \min_k n_k \ge \kappa' n$ for some constant $\kappa' >0$. Also define $n_{\max} = \max_{k}n_k$.
  \end{ass}


  \begin{thm}[Consistency of the Right SC algorithm]\label{thm:CommunityConsistency}
Assume the network $\tilde{A}$ is generated from model \eqref{eq:generic-model}.  Let  $\hat{\V{c}}$ be the output of the Right SC algorithm with a $(1+\epsilon)$-optimal solution, $\sigma_K(\tilde{P}) $ be  the $K$th largest singular value of $\tilde{P}$, and $\norm{\tilde{P}}_{\infty} = \max_{ij} \tilde{P}_{ij}$. Assume \ref{A2} holds, and $n\norm{\tilde{P}}_{\infty} \ge C_0 \log{n}$ for some constant $C_0$. If there exists a constant $C_1$ depending on $C_0$, $\epsilon$ and $\kappa'$, such that
\begin{equation}\label{eq:RightSC-basic-signal}
  \frac{Kn\norm{\tilde{P}}_{\infty} }{\sigma_K(\tilde{P})^2} \le \frac{1}{C_1},
  \end{equation}
   then with probability at least $1-n^{-1}$, there exists a permutation of labels $\hat{\V{c}}$, such that 
   $$\sum_k \frac{|G_k \setminus \hat{G}_k|}{n_k} \le C_1\frac{Kn\norm{\tilde{P}}_{\infty} }{\sigma_K(\tilde{P})^2}.$$
     \end{thm}

Theorem~\ref{thm:CommunityConsistency} implies that under the general model \eqref{eq:generic-model}, the Right SC only misclusters a proportion of nodes upper bounded by the quantity $\frac{Kn\norm{\tilde{P}}_{\infty} }{\sigma_K(\tilde{P})^2} \to 0$. Unfortunately, due to the generality of model \eqref{eq:generic-model},  the bound of Theorem~\ref{thm:CommunityConsistency} depends on $\sigma_K(\tilde{P})$, a quantity without an obvious interpretation.   Under the NSBM, this result can be simplified into a more interpretable form.

 \begin{ass}[Simplified Parameterization of the NSBM]\label{A3}
   Assume matrix $B$ is a fixed matrix. Write $\theta_i = \rho_n\bar{\theta}_i$ where $\bar{\theta}_i$'s are independently sampled from a fixed  discrete   distribution $g_{\theta}$ on $m_1$ different positive values.
   Assume the values of $\lambda_i$'s are independently sampled from a fixed multinoulli distribution $g_{\lambda}$ with mean value 1 on $m_2$ different  values and then rescaled to satisfy the identifiability constraints in Proposition~\ref{prop:identify}.    
  \end{ass}
 
This parameterization ensures that $\rho_n$ is the only quantity varying with $n$.  Setting $\theta_i = \rho_n\bar{\theta}_i$ allows us to parameterize the edge density of the network by a single parameter $\rho_n$. In fact, under the NSBM model with identifiability constraints of Proposition~\ref{prop:identify}, \ref{A2} and \ref{A3}, it is easy to see that $n\rho_n$ is the order of the expected average node degree.  
Then we have the following corollary of Theorem~\ref{thm:CommunityConsistency}. 
 
 \begin{coro}\label{coro:CommunityConsistency}
Assume the network is generated from the NSBM \eqref{eq:NSBM}. Let $\hat{\V{c}}$ be the clustering labels found by the Right SC algorithm with $(1+\epsilon)$-optimal solution. If assumptions of Proposition \ref{prop:identify}, \ref{A2}, and \ref{A3} hold and $n \rho_n \ge C_0 \log n$ for some constant $C_0$, then for sufficiently large $n$, with probability at least $1-2n^{-1}$, there exists a permutation of labels $\hat{\V{c}}$, such that 
\begin{equation}\label{eq:community}
\sum_k \frac{|G_k \setminus \hat{G}_k|}{n_k} \le C'\frac{1}{n\rho_n}
\end{equation}
   for some constant $C'$ depending on $C_0, \kappa',\epsilon, \eta, K$ and the distributions of $\bar{\theta}_i$'s and $\lambda_i$'s.
 \end{coro}

 The Corollary~\ref{coro:CommunityConsistency} states that as long as the expected average degree of the network $n\rho_n $ grows at least in the order of $\log{n}$, the mis-clustered proportion is bounded by the order of $1/n\rho_n$.

  Next, we introduce the consistency of the Right SMST (Algorithm~\ref{algo:RightSMST}). The strong consistency can be obtained by using the recently $\ell_\infty$ perturbation theory from \cite{lei2019unified}. 
  
    \begin{thm}[Consistency of the Right SMST algorithm]\label{thm:CommunityConsistency-mst}
Assume the network $\tilde{A}$ is generated from the general model \eqref{eq:generic-model}.  Let  $\hat{\V{c}}$ be the output of Algorithm~\ref{algo:RightSMST}, and $\norm{\tilde{P}}_{\infty} = \max_{ij} \tilde{P}_{ij}$. Assume \ref{A2} holds,  $n\norm{\tilde{P}}_{\infty} \ge C_0 \log{n}$,  and
 \begin{equation}\label{eq:generic-basic-signal}
   \sigma_{K}(\tilde{P}) \ge C_1 n\norm{\tilde{P}}_{\infty}
  \end{equation}
   for some constants $C_0, C_1>0$.  If the following condition \eqref{eq:generic-strong-signal} is true
   \begin{equation}\label{eq:generic-strong-signal}
   \sqrt{\frac{\log{n}}{n\norm{\tilde{P}}_{\infty}}}\max\left( \sqrt{n}\norm{U}_{2,\infty}, \sqrt{n}\norm{V}_{2,\infty}\right) = o(1),
   \end{equation}
   for sufficiently large $n$, then there exists a permutation $\Psi: [K] \to [K]$ of community labels such that 
   $$\Psi(\hat{\V{c}}) = \V{c}$$
   with probability at least $1-n^{-1}$.
     \end{thm}
  
Compared to Theorem~\ref{thm:CommunityConsistency} for the Right SC, Theorem~\ref{thm:CommunityConsistency-mst} requires one additional condition \eqref{eq:generic-strong-signal} to achieve strong consistency.   Condition \eqref{eq:generic-strong-signal} reduces to \eqref{eq:generic-basic-signal} only  when $\max\left( \norm{U}_{2,\infty}, \norm{V}_{2,\infty}\right) = O(\frac{1}{\sqrt{n}})$, which also means that  $\tilde{P}$ has perfect incoherence \citep{candes2009exact,candes2010power}.  Again, we give a simplified form of Theorem~\ref{thm:CommunityConsistency-mst} in the special case of the NSBM with the parameterization assumed in \ref{A3}.

      \begin{coro}[Consistency of the Right SMST algorithm under NSBM]\label{coro:CommunityConsistency-mst-NSBM}
Assume the network $\tilde{A}$ is generated from the NSBM \eqref{eq:NSBM}.  Let  $\hat{\V{c}}$ be the output of Algorithm~\ref{algo:RightSMST}. If assumptions of Proposition \ref{prop:identify}, \ref{A2}, and \ref{A3} hold, and 
$$n\rho_n/\log{n} \to \infty, $$then for sufficiently large $n$, with probability at least $1-n^{-1}$, there exists a permutation $\pi: [K] \to [K]$ of labels $\hat{\V{c}}$ such that 
   $$\pi(\hat{\V{c}}) = \V{c}.$$
     \end{coro}

  \subsection{Parameter estimation consistency under the NSBM}\label{secsec:PETheory}
   
As usual, consistency of parameter estimation follows from consistency of labels.   In  this section, we show consistency for parameters $B, \V{\lambda}$ and $\V{\theta}$ estimated by Algorithm~\ref{algo:PE} under the NSBM \eqref{eq:NSBM} with parameterization \ref{A3}.  Since we have  already shown that the community labels can be recovered, in this section we treat the community labels as known for simplicity.

  \begin{thm}\label{thm:IdeaEstimationErrorBound}
Assume the network is generated from the NSBM \eqref{eq:NSBM}.	Let $\hat{\V{\theta}}, \hat{\V{\lambda}}$ and $\hat{B}$ be the estimators for $\V{\theta}, \V{\lambda}$ and $B$, respectively, obtained by Algorithm~\ref{algo:PE}. 
  Assume conditions of Proposition \ref{prop:identify}, \ref{A2} and \ref{A3} hold. Then there exists a constant $c_1$, depending on $\kappa'$ in \ref{A2}, $ B, g_{\lambda}$ in \ref{A3}, and $K$, such that if $\rho_n \ge  c_1\frac{\log^4{n}}{n}$,  for sufficiently large $n$ we have
   \begin{align*}
\max_i \frac{|\hat{\theta}_i-\theta_i|}{\theta_i} & \le \frac{1}{\log{n}} \, , \\
\max_{k,l} |\hat{B}_{kl} - B_{kl}| &\le \frac{1}{\log{n}} \, , \\
\max_{i} |\hat{\lambda}_{i} - \lambda_{i}| &\le \frac{1}{\log{n}}
\end{align*}
with probability at least $1 - n^{-1}$.
  \end{thm}
  
  Theorem~\ref{thm:IdeaEstimationErrorBound} shows that the estimators are consistent even on the sparse networks, with average degree on the order of $\log^4{n}$. This requirement can be slightly improved at the cost of slightly worse probability bound or the upper bounds on the errors.

 \section{Numerical results on synthetic networks}\label{sec:sim}
 
In this section, we demonstrate the proposed methods for community detection and network modeling using simulation studies. We first show the importance of clustering based on the correct spectral information by comparing the performance of the Right SC with a few other spectral methods that do not consider the nature of the edge nomination mechanism.  We also show that the NSBM cannot be approximated well by a few standard community models.   We illustrate this on both unweighted networks and weighted networks with Poisson-distributed edge weights, since both our community detection algorithm and the parameter estimation algorithm by the method of moments apply to both settings.

 \subsection{Evaluation on binary networks}
 
For this set of experiments, networks are generated from the NSBM as follows: $n=1200$ nodes are randomly assigned to $K=3$ communities with equal probability. The matrix $B$ has all diagonal elements equal to 1 and all off-diagonal elements equal to $\beta$.  The parameters $\lambda_i$'s are generated independently with $\log(\lambda)$  sampled uniformly from the interval $(-t,t)$, and then rescaled to satisfy the constraint $\sum_{\V{c}_i = k}\lambda_i = n_k$ for each $k$. Each $\theta_i$ is  independently set either to $c$ or $0.05 c$, with probability 0.5 each, with the value of $c$ chosen so that the resulting average degree of the network is 50. 

We evaluate several spectral clustering algorithms for community detection to demonstrate the importance of identifying the informative part of the eigenstructure rather than blindly relying on spectral clustering.   For a directed network, one alternative to the Right SC algorithm we proposed is to use the left singular vectors for clustering,  which we call ``Left SC''. However, it is clear that under NSBM Left SC will fail given it does not take node heterogeneity into account.   Therefore we instead consider left spherical SC (Left SSC), which first normalizes each row of the matrix of left singular vectors before applying the $K$-means algorithm as a standard way to deal with row heterogeneity \citep{lei2014consistency}.  
We do not include the spherical version of Right SC, since they both use the right singular vectors and produce similar results.   Essentially, since there is no column heterogeneity under this model, there is no need for column normalization.

A common approach to community detection for directed networks is to convert $\tilde{A}$ to a symmetric matrix and then apply an algorithm for an undirected network.    This is typically accomplished by connecting two nodes  in the undirected network if there is an edge in either direction.    Applying SC and SSC to the symmetrized network gives two more options, ``Symmetric SC'' and ``Symmetric SSC'', but they again give similar results, and thus, we omit the spherical version.
This strategy is equivalent to treating the network as generated from the SBM or the degree-corrected SBM \citep{karrer2011stochastic}, respectively.

Lastly, we investigate other models for edge nomination.    Specifically, we include the Bayesian method from \cite{butts2003network}, which includes the model of \cite{newman2018networka} as a special case, assuming that the nomination process does not depend on the individual or the connections strength. This model is representative of the current literature on modeling missing links. The model, however, is not designed for community detection.  Therefore, we take the posterior mean of the probability matrix as input for spectral clustering, which again results in ``Left", ``Symmetric" and "Right" versions.  The posterior inference is implemented in R package \texttt{sna} \citep{sna} and we refer to the three versions as ``Bayesian-Left", ``Bayesian-Symmetric", and ``Bayesian-Right".  This method is computationally expensive and took a very long time to run;  it cannot be applied to large networks.

We evaluate the community detection performance by using the cluster accuracy, defined to be mis-clustered proportion, under the best permutation within the $K$ labels of the estimated clusters. In addition to community detection accuracy, we also evaluated the estimation error of $\tilde{P}$. Though using the true model (the NSBM) is expected to give the best results, we are interested in how closely they can be approximated by using simpler models.    Therefore,  we compared results of Algorithm \ref{algo:PE} to three other computationally feasible network models with communities: the directed SBM and its degree-corrected version, and the stochastic co-clustering block model (SCBM) of \cite{rohe2016co}. The SCBM is also based on the idea of directed DCSBM but assumes different community memberships for senders and receivers; therefore, it is not directly comparable on community detection. However, it can also provide an approximation to $\tilde{P}$.  Lastly, the Bayesian method of \cite{butts2003network} also gives full posterior network distribution, and we take the posterior mean of $P$ together with the edge flipping probability to construct an estimate of $\tilde{P}$. The estimation accuracy is measured by the relative error $\norm{\tilde{P} - \hat{P}}_F^2/\norm{\tilde{P}}_F^2$ averaged over 100 replications, where $\hat{P}$ is the estimated probability matrix.  

\begin{figure}[H]
\centering
\begin{subfigure}[t]{0.48\textwidth}
\centering
\includegraphics[width=\textwidth]{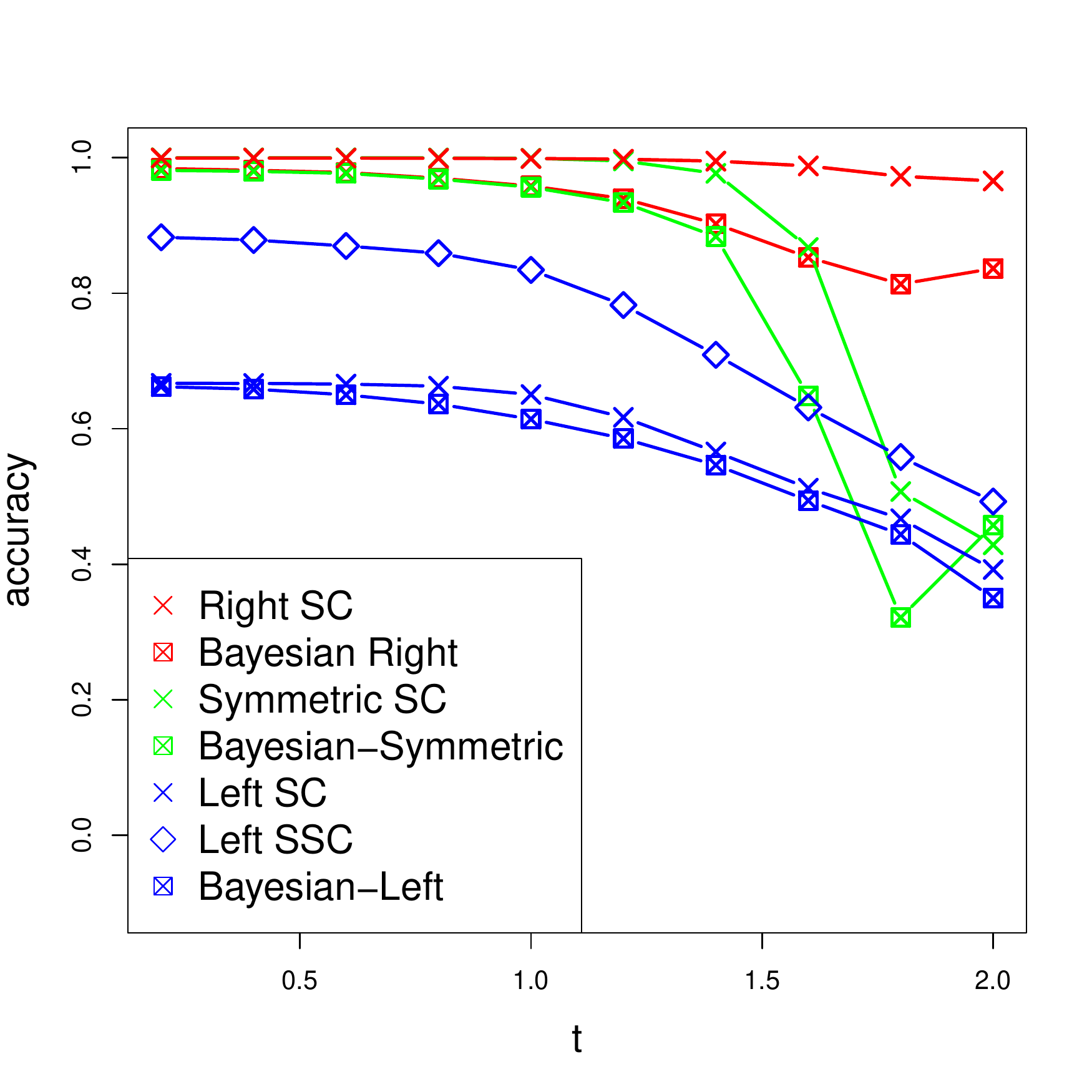}
\caption{Community detection accuracy.}
\label{fig:Varyt-binary-Acc}
\end{subfigure}
\hfill
\begin{subfigure}[t]{0.48\textwidth}
\centering
\includegraphics[width=\textwidth]{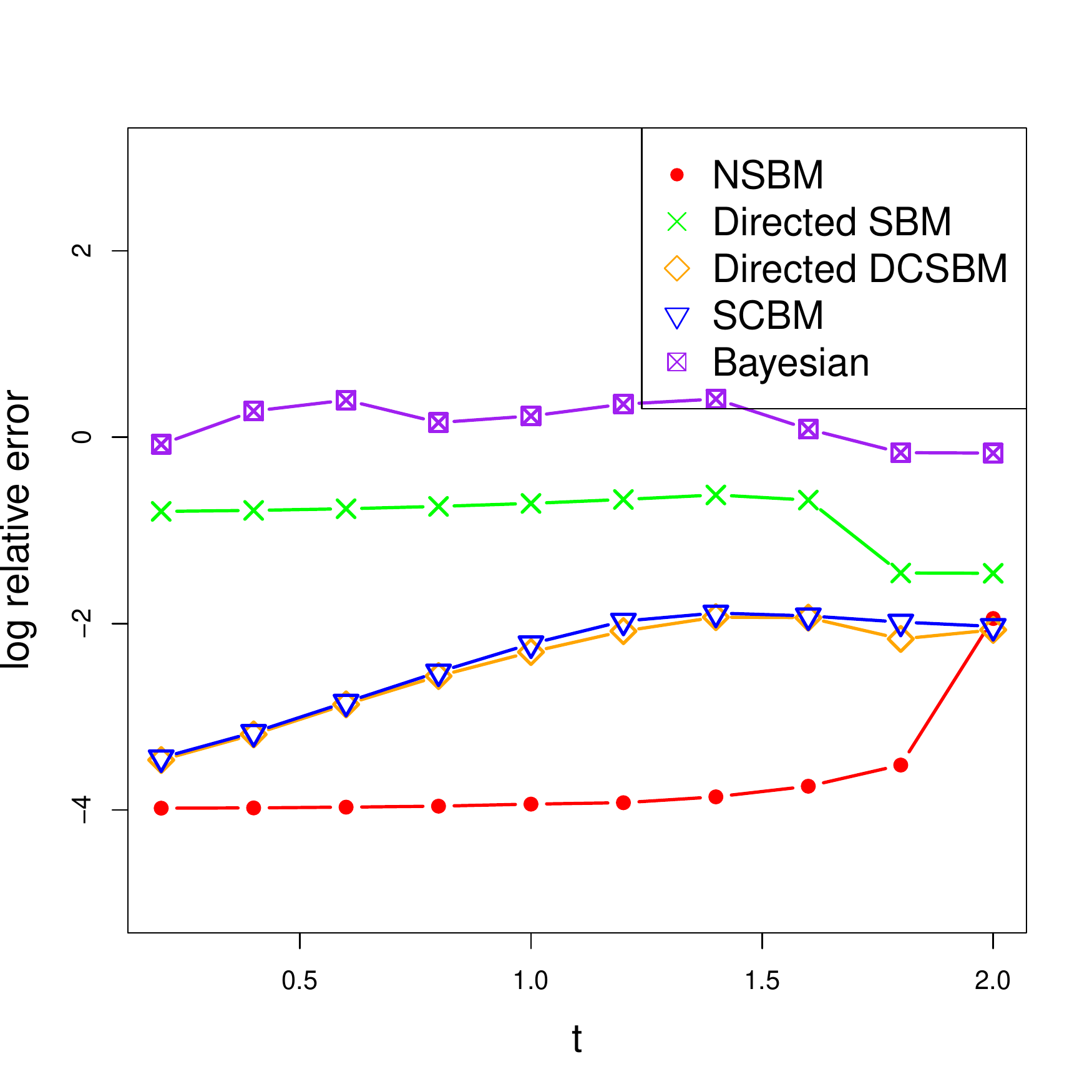}
\caption{Estimation error of $\tilde{P}$ (log scale).}
\label{fig:Varyt-binary-Err}
\end{subfigure}
\caption{Community detection accuracy and probability matrix estimation error when the network is generated from the unweighted NSBM with $\beta = 0.2$ and varying $t$. }
\label{fig:Varyt-binary}
\end{figure}

We start from varying $t$ from 0.2 to 2 while keeping $\beta = 0.2$ fixed.  The results are shown in Figure~\ref{fig:Varyt-binary}.   For community detection, all methods based on the right singular vectors are better than their counterparts using the other types of spectral structures.  In particular, when  $t$ is small, the probability of nomination does not depend on the connection strength that much, and thus spectral clustering based on the standard SBM (or DCSBM)  still works.   As $t$ increases and the nomination mechanism becomes more heterogeneous across the nodes, symmetric clustering methods fail. The Left SSC is even worse since it relies entirely on the senders' information, where the community structure is masked by heterogeneity of nominations. The Bayesian method is not effective in removing the impact of the edge nomination effects.
For estimating the probability matrix,  the NSBM unsurprisingly works the best because it uses the correct model. More importantly, even for small  values of $t$ where symmetric methods perform ok on community detection, none of the other methods come close on estimating the probability matrix.  Moreover, the estimation algorithm for the NSBM remains stable for most of the $t$ in the range, only beginning to degrade when the clustering accuracy drops.  


Next, we compare different methods while varying the signal-to-noise ratio, to see whether the effects of edge nomination become negligible when communities are well separated. Specifically, we vary $\beta$ from 0.1 to 0.9 and fix $t=1.5$. The corresponding results are shown in Figure~\ref{fig:Varybeta-binary}.    The Right SC retains a consistent advantage over the entire range of $\beta$, though all methods fail to give informative clustering for $\beta \ge 0.7$. For model estimation, our method is better than the other for $\beta \le 0.6$. For even larger $\beta$, because of the fading performance of clustering and the higher model complexity, the model estimation error becomes higher than the simpler SBM and SCBM.  Notably, the directed SBM and SCBM slightly improve on estimation as $\beta$ increases, possibly because the probabilities become more homogeneous and easier for the SCBM (or the directed SBM) to approximate. 

\begin{figure}[H]
\centering
\begin{subfigure}[t]{0.48\textwidth}
\centering
\includegraphics[width=\textwidth]{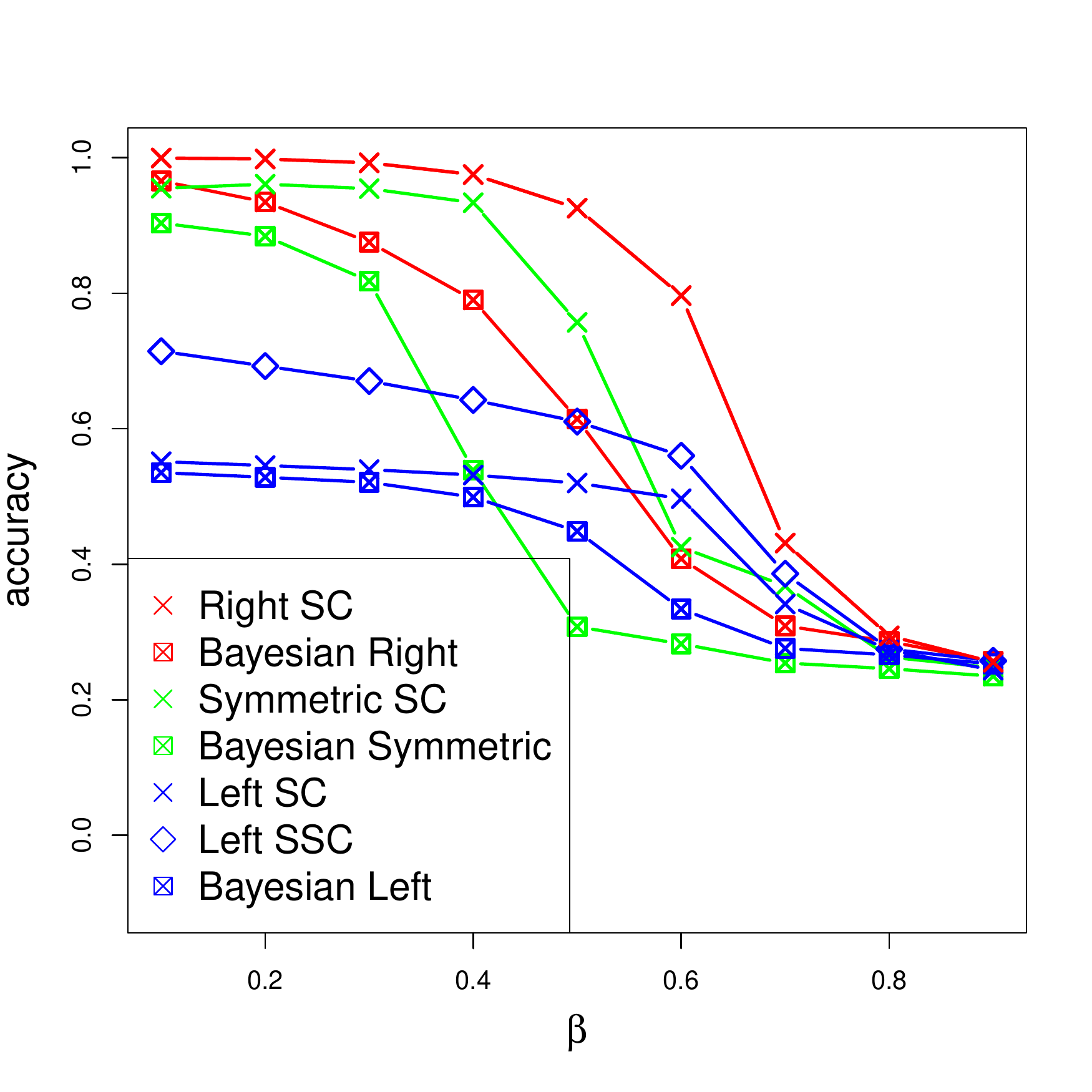}
\caption{Community detection accuracy.}
\label{fig:Varyt-binary-Acc}
\end{subfigure}
\hfill
\begin{subfigure}[t]{0.48\textwidth}
\centering
\includegraphics[width=\textwidth]{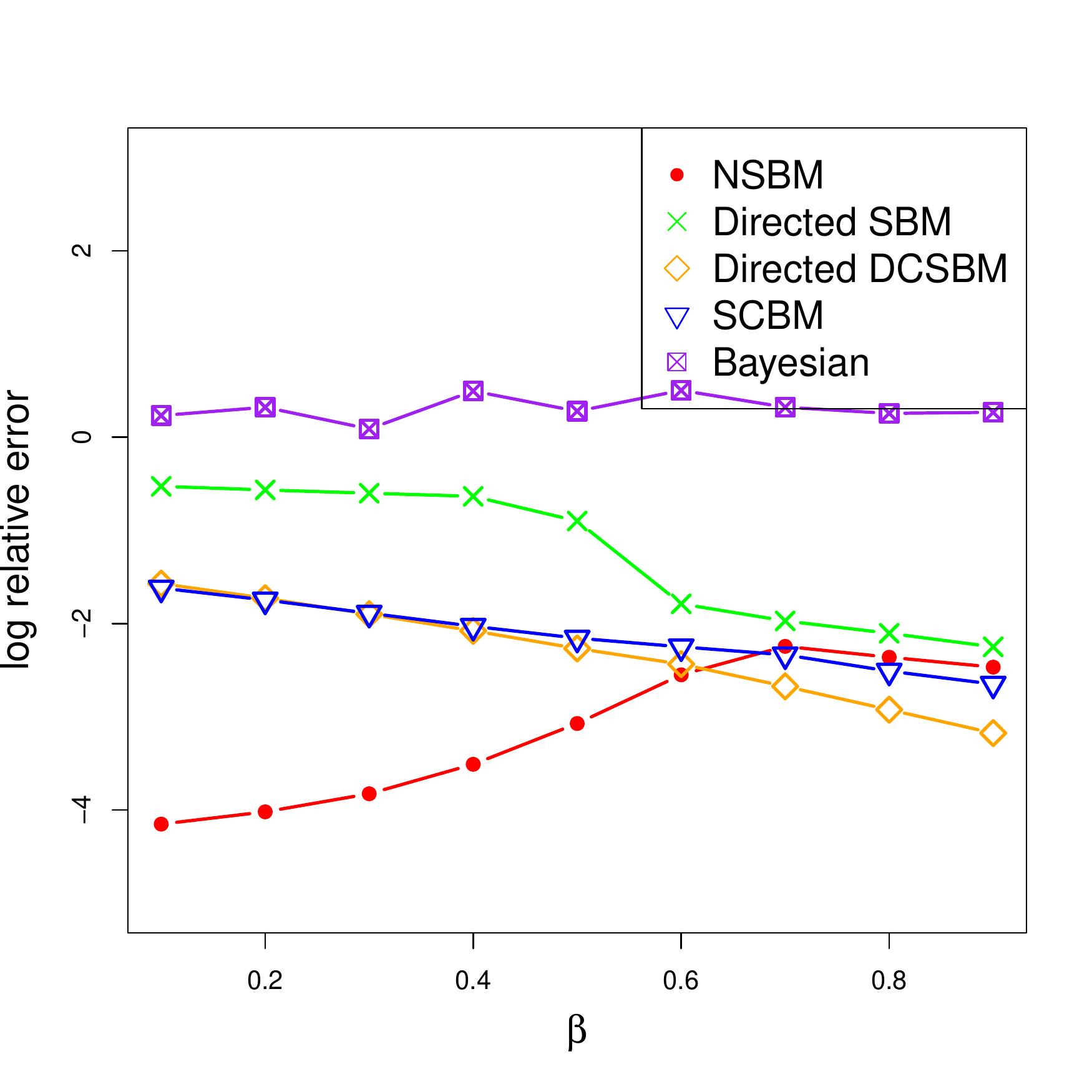}
\caption{Estimation error of $\tilde{P}$ (log scale).}
\label{fig:Varyt-binary-Err}
\end{subfigure}
\caption{Community detection accuracy and probability matrix estimation error when the network is generated from the unweighted NSBM with $t = 1.5$ and varying $\beta$. }
\label{fig:Varybeta-binary}
\end{figure}

%

\subsection{Evaluation on weighted networks}

For this set of experiments, we generate the networks from a Poisson distribution with the same NSBM structure. The Bayesian method from the previous section is no longer applicable, but the other methods can still be used. We use the same two performance metrics and show that our method can be applied to weighted networks without any changes and has similar advantages over its competitors to what we observed in the previous section.
The only difference from the settings described above is that the value of $c$ is set so that the average row sum of $\tilde{A}$ is 250, which roughly gives an average of $50$ nonzero entries in each row of $\tilde{A}$, matching the degree of the unweighted networks.

\begin{figure}[H]
\centering
\begin{subfigure}[t]{0.48\textwidth}
\centering
\includegraphics[width=\textwidth]{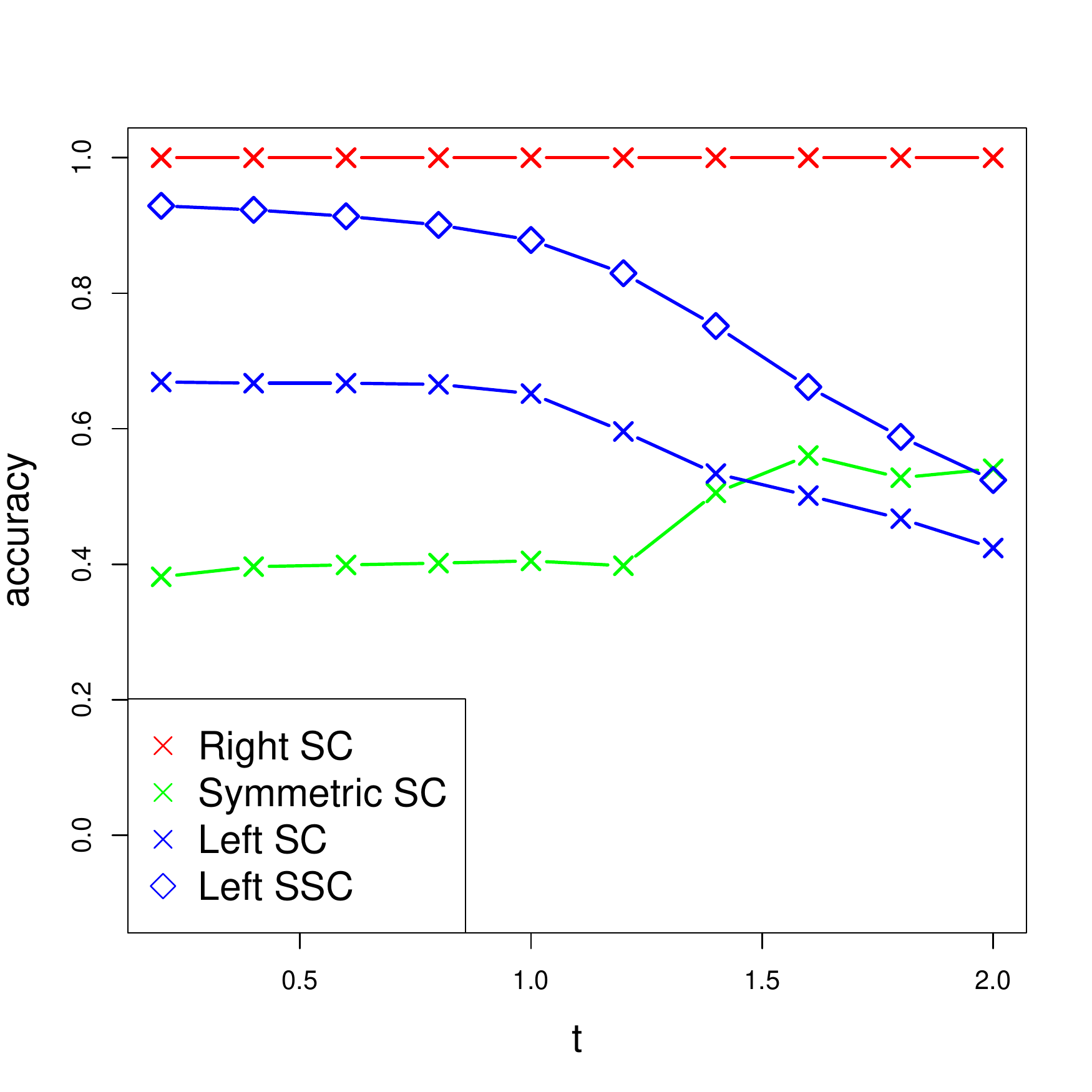}
\caption{Community detection accuracy.}
\label{fig:Varyt-binary-Acc}
\end{subfigure}
\hfill
\begin{subfigure}[t]{0.48\textwidth}
\centering
\includegraphics[width=\textwidth]{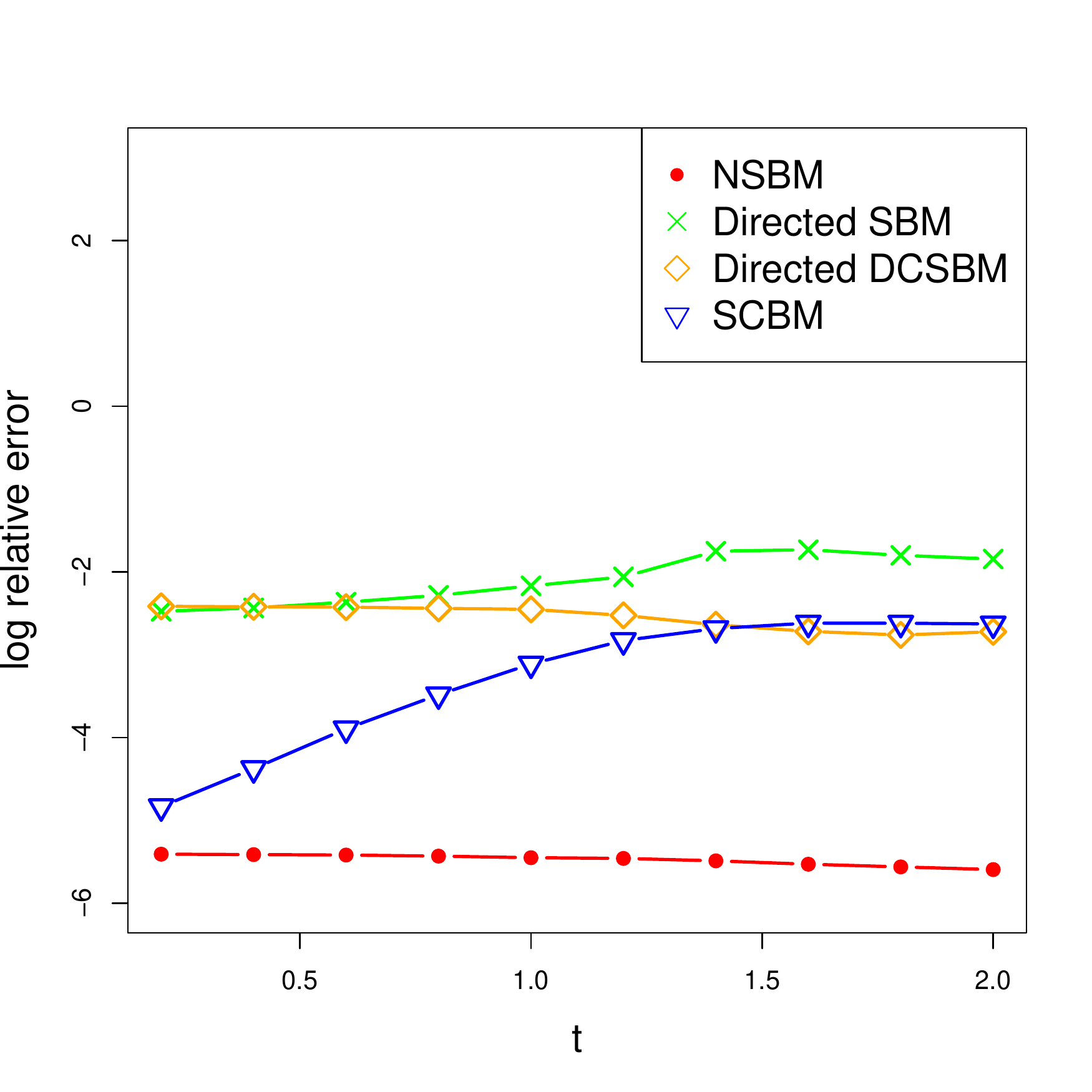}
\caption{Estimation error of $\tilde{P}$ (log scale).}
\label{fig:Varyt-binary-Err}
\end{subfigure}
\caption{Community detection accuracy and probability matrix estimation error when the network is generated from NSBM with Poisson edge weights,  with $\beta = 0.2$ and varying $t$. }
\label{fig:Varyt-weighted}
\end{figure}

The results are shown in Figures~\ref{fig:Varyt-weighted} and~\ref{fig:Varybeta-weighted}.   Our method retains the advantages it had on unweighted networks.  Community detection becomes easier in this case, because Poisson distribution is more informative compared with Bernoulli. The Right SC remains accurate for the range of $t$ from 0 to 2. The method is also very stable for most values of $\beta$, only deteriorating around $\beta > 0.8$.

\begin{figure}[H]
\centering
\begin{subfigure}[t]{0.48\textwidth}
\centering
\includegraphics[width=\textwidth]{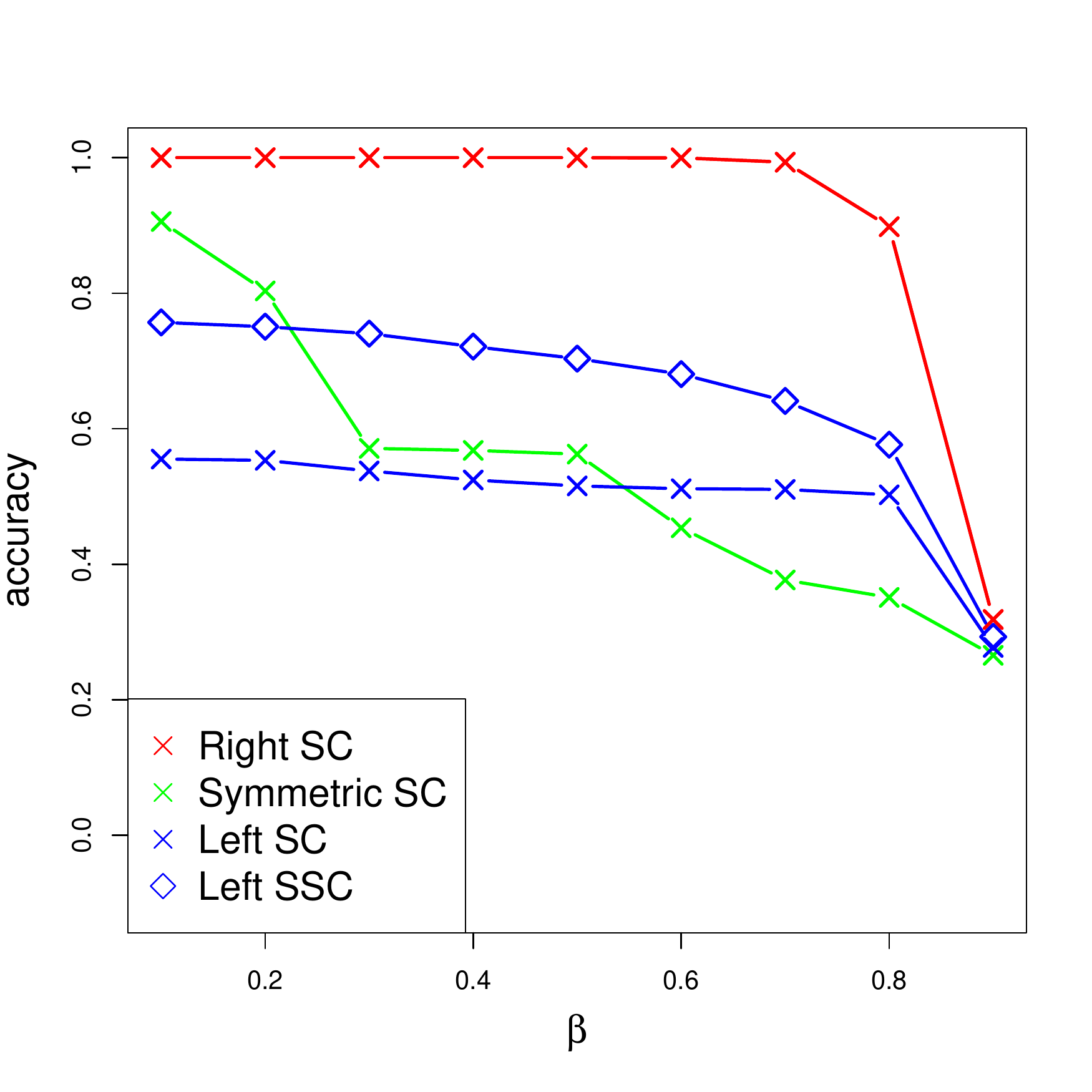}
\caption{Community detection accuracy.}
\label{fig:Varyt-binary-Acc}
\end{subfigure}
\hfill
\begin{subfigure}[t]{0.48\textwidth}
\centering
\includegraphics[width=\textwidth]{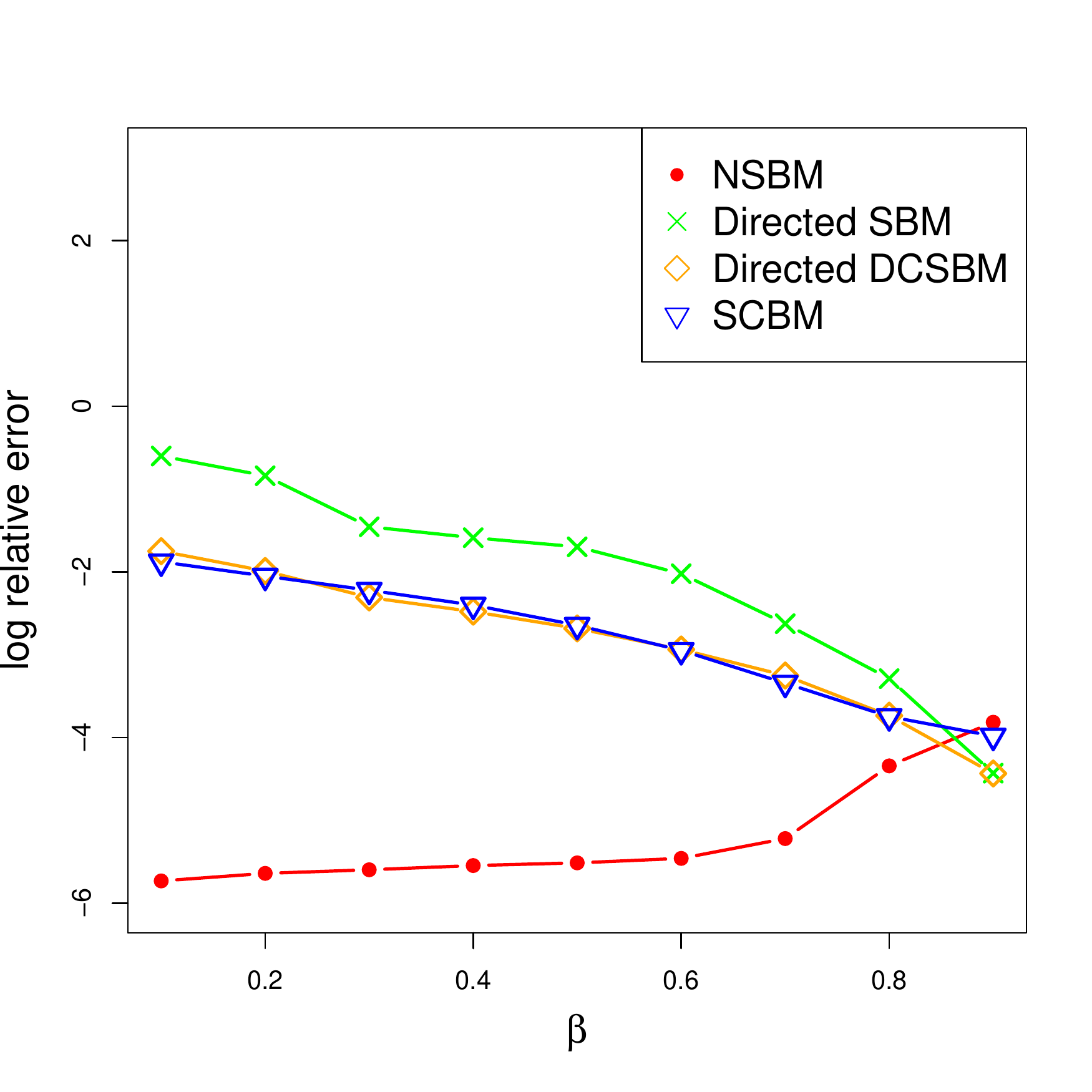}
\caption{Estimation error of $\tilde{P}$ (log scale).}
\label{fig:Varyt-binary-Err}
\end{subfigure}
\caption{Community detection accuracy and probability matrix estimation error when the network is generated from NSBM with Poisson edge weights,  with $t = 1.5$ and varying $\beta$. }
\label{fig:Varybeta-weighted}
\end{figure}

\section{Business school faculty hiring network analysis}\label{sec:BSchool}

Here we apply the proposed NSBM to a faculty hiring network between US Business schools.  The data were collected by \cite{clauset2015systematic} via web crawling and records on 18,924 tenure or tenure-track faculty members, recording the institution from which they obtained their Ph.D.\ and the institution by which they were hired. The original dataset covers faculty in business, computer science, and history; we analyzed the business school data only. 


\begin{figure}[H]
\begin{center}
\includegraphics[scale=0.38]{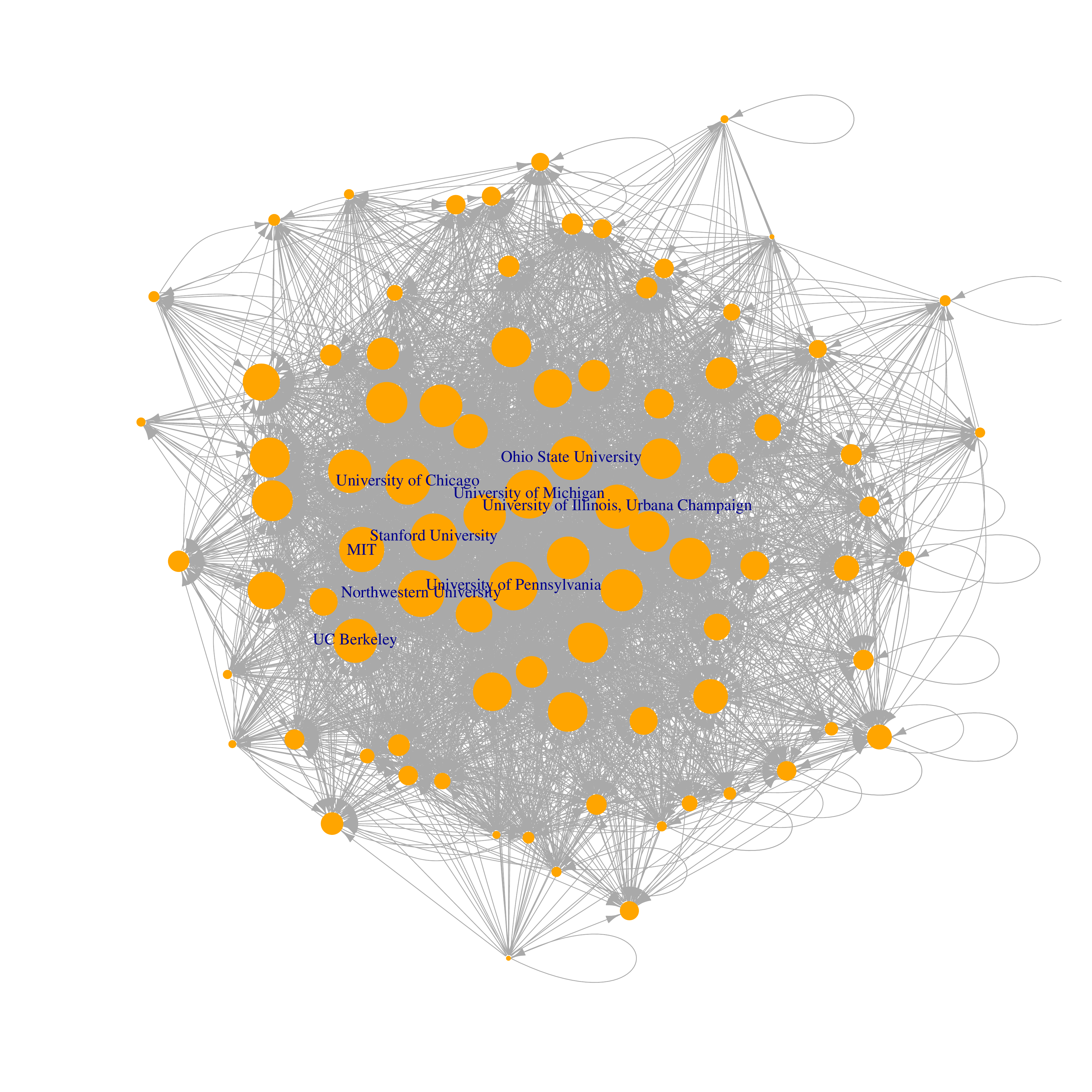}
\vspace{-1cm}
\caption{The hiring network between 87 U.S. business schools. An edge from $i$ to $j$ indicates that institution $i$ has hired Ph.D. graduates from institution $j$. The node size is proportional to the number of incoming edges.}
\label{fig:BusinessNetwork}
 \end{center}
\end{figure}

The business school data covers 7856 faculty members from 112 institutions. To reduce noise, we removed institutions with either receiver or sender degree of 3 or less, resulting in 87 institutions remaining. The hiring numbers have a heavy-tailed distribution, and we applied truncation for stability, as described in Section~\ref{sec:extent}. Specifically, we construct a network by creating an edge  from $i$ to $j$ with weight 1 if institution $i$ has hired exactly one faculty member with a Ph.D. from institution $j$.  If institution $i$ has hired more than one graduate of institution $j$, we set the edge weight to 2.    We found that this truncation of weights does improve stability, whereas setting all edge weights to 0 or 1 loses too much information.   
The resulting directed network with 87 nodes is shown in Figure~\ref{fig:BusinessNetwork}, where the node size is proportional to the receiver degree, i.e., the number of institutions to which institution $i$ has sent its graduates.
We are interested in finding communities of institutions as well as investigating whether there are hiring ``inequalities'' between these communities.     The NSBM suits the hiring network well,  because we do not observe job offers that did not result in a hire, and we can assume that most institutions made some offers that were declined.   It is also safe to assume we do not observe any false edges.

To determine the number of communities, we applied the edge cross-validation method with average stability selection proposed by \cite{li2016network}.  The procedure suggests $K=4$ communities for this dataset.  We then fit the NSBM to the network with $K= 4$. Table~\ref{tab:BSchoolCommunity} shows the communities and their average rankings from two sources, the US NEWS graduate school rankings from 2012 (included in the data set)  and the $\pi$-ranking proposed by \cite{clauset2015systematic}.  The $\pi$-ranking is designed to measure hiring advantage, with a higher-ranked institution tending to be more successful in hiring competitive candidates.  We list up to 15 institutions (the first two communities only have 12 each, whereas the other two have 19 and 44) with the highest $\pi$-ranking within each community in Table \ref{tab:BSchoolCommunity} .  Overall, the communities match both rankings very well, showing a clear ordering, with the first community mostly consisting of top business schools, the second one of good but slightly lower-ranked schools, and so on.  

\begin{table}[H]
{\centering
\caption{Communities of business schools found by NSBM, their average and median rankings from US News 2012 and $\pi$-ranking of \cite{clauset2015systematic}.  Up to 15 institutions with the highest $\pi$-ranking are listed for each community.}
\label{tab:BSchoolCommunity}
\small
\begin{tabular}{|p{0.3cm}|p{0.5cm}|p{1.8cm}|p{1.8cm}|p{9cm}|}

  
  \hline
 & size & USNews (avg./med.) & $\pi$-ranking  (avg./med.) & Institutions \\ 
  \hline
  1 & 12&7.7/8& 8.3/8&Stanford, MIT, Harvard, UC Berkeley, U Chicago, Cornell, U Michigan, Columbia, Yale, U Penn., NYU, Duke\\ 
  \hline
  2 & 12&29.8/32.5 &17.7/17.5 &U Rochester, Northwestern, Carnegie Mellon, U Wisconsin Madison, UCLA, U Minnesota-Twin Cities, UIUC, Purdue, U Florida, UT Austin,  U Washington \\ 
   \hline
   3&19& 53.1/54&45/45& Ohio State, UNC Chapel Hill, U Pittsburgh, Penn.\ State, Indiana U., Michigan State, Georgia Tech, U Arizona, SUNY Buffalo, Texas A\&M, U Georgia, Arizona State, U South Carolina, Virginia Tech, Florida State
   \\
   \hline
   4&44& 63.7/63 & 61.4/61.5 & Washington U St.\ Louis, U Maryland College Park, U Colorado Boulder, UC Irvine, U Utah, U Oregon, U Southern California, UT Dallas, U Virginia, Boston U., UMass Amherst, Emory, Case Western, UC Davis, Vanderbilt \\
   \hline
\end{tabular}
}
\end{table}

The parameters of NSBM can be directly interpreted in terms of a hierarchy in hiring, which was reported by \cite{clauset2015systematic}. Based on the weighted NSBM in Section~\ref{sec:extent}, we define connection strength from community $k$ to community $l$  as the expectation of average edge weights from nodes $i\in G_k$ to nodes $ j \in G_l$, 
$$M_{kl} = \frac{1}{n_kn_l}\sum_{i \in G_k, j\in G_l}\theta_iB_{ij}^{\lambda_i} .  $$

\begin{table}[H]
\caption{Estimated strengths of connections between communities.} 
\label{tab:Bmatrix}
\centering
\begin{tabular}{|l|cccc|}
  \hline
 & Group 1 & Group 2 & Group 3 & Group 4 \\ 
  \hline
Group 1 & 1.86 & 0.93 & 0.15 & 0.06 \\ 
  Group 2 & 1.56 & 1.31 & 0.42 & 0.13 \\ 
  Group 3 & 0.86 & 1.29 & 0.98 & 0.32 \\ 
  Group 4 & 0.76 & 0.86 & 0.51 & 0.22 \\ 
  \hline
\end{tabular}
\end{table}

Table~\ref{tab:Bmatrix} shows the estimated connection strengths for the business hiring network.  It shows that Group 1 institutions tend to hire the most from their own group, and about half as many from Group 2.  They are not very likely to hire from Groups 3 and 4. Group 2 institutions hire roughly equally from Groups 1 and 2, and a fraction from Group 3, but very few from Group 4. Group 3 institutions hire the most from Group 2, not Group 1.  Group 4 follows a similar pattern, hiring more from groups closer to itself. The estimated model parameters thus indicate a strong hierarchy in hiring relationships between the groups, which aligns closely with the rankings in Table~\ref{tab:BSchoolCommunity}.

In addition to community parameters,  NSBM allows us to estimate the hiring preferences of individual institutions as represented by parameters $\lambda_i$. Table~\ref{tab:lambda} shows the estimated $\lambda_i$'s  for Group 1, indicating how strongly each institution follows the community level preferences.   For instance, we see that Yale and Cornell show a stronger  preference for hiring graduates from their own group,  while the University of Michigan and the University of Pennsylvania are not so stringent. 

\begin{table}[H]
\caption{Estimated $\lambda_i$'s for Group 1 institutions.} 
\label{tab:lambda}
\centering
\begin{tabular}{|l|c|c|c|}
  \hline
 Institution & $\hat{\lambda}_i$ & USN ranking & $\pi$-ranking \\ 
  \hline
Yale  & 2.17 & 10 &  11 \\ 
Cornell  & 1.88 & 16 &   7 \\ 
Columbia  & 1.10 & 9 &  10 \\ 
Harvard  & 0.97 & 2 &   3 \\ 
MIT & 0.96 & 3 &   2 \\ 
UC Berkeley & 0.85 & 7 &   4 \\ 
U of Chicago & 0.75 & 5 &   6 \\ 
Stanford  & 0.74 & 1 &   1 \\ 
New York U & 0.70 & 10 &  16 \\ 
Duke  & 0.69 & 12 &  19 \\ 
U  Michigan & 0.61 & 14 &   9 \\ 
U  Pennsylvania & 0.57 & 3 &  12 \\ 
   \hline
\end{tabular}
\end{table}

Overall, fitting the NSBM to the faculty hiring network revealed a clear hierarchical structure in the hiring relationships between institutions, which matches both our expectations and the observations of \cite{clauset2015systematic}.

For comparison, we briefly discuss community detection results on this network obtained by spectral clustering applied to the undirected version. The four communities are shown in Table~\ref{tab:SymmetricBSchoolCommunity}, with their average and median ranking by US News and $\pi$-ranking, and 15 institutions with the highest $\pi$-ranking in each community.  
The first group is still higher-ranked even though it no longer includes universities like Yale, Cornell, and Columbia.  The other  three groups, however, all have similar average rankings and no discernible interpretation that we could think of that might result in such groupings.   The striking difference between the average rankings of groups from these two clustering results confirms the importance of accounting for the nomination mechanism and using the correct spectral information.

\begin{table}[H]
{
\centering
\caption{Communities of business schools found by symmetric spectral clustering, their average and median rankings from US News 2012 and $\pi$-ranking of \cite{clauset2015systematic}.  Up to 15 institutions with the highest $\pi$-ranking are listed for each community.}
\label{tab:SymmetricBSchoolCommunity}
\small
\begin{tabular}{|p{0.3cm}|p{0.5cm}|p{1.8cm}|p{1.8cm}|p{9cm}|}
  \hline
 & size & USN (avg./med.) & $\pi$-ranking  (avg./med.) & Institutions \\ 
  \hline
  1 & 19&19.2/14& 17.8/13&Stanford, MIT, Harvard, UC Berkeley, U Rochester, U Chicago, Northwestern, U Michigan, U Penn., Carnegie Mellon, NYU, U Minnesota Twin Cities, Duke, UNC Chapel Hill, U Washington St. Louis 
  \\ 
   \hline
  2&20& 55.1/56.5 & 44.6/42 &Cornell, Columbia, U Wisconsin-Madison, UIUC, Ohio State, U Florida, U Pittsburgh, Penn State, Michigan State, SUNY Buffalo, U Mass Amherst, Syracuse, Tulane, U Connecticut, U Cincinnati 
  \\
   \hline
  3 & 24&52.7/40 &54/49 &Yale, UCLA, U Washington, U Colorado Boulder, UC Irvine, U Utah, U Oregon,  UT Dallas, U Virginia, Boston U,  UC Davis, Vanderbilt, Claremont Graduate U, U Houston, Rice U 
  \\ 
   \hline
  4&24& 63.8/63&56/56.5& Purdue, U Iowa, UT Austin, Indiana U, Georgia Tech, U Arizona, Texas A\&M, U Georgia, Arizona State, U South Carolina, Virginia Tech, Florida State, U Oklahoma, U Kansas, Louisiana State 
  \\
  \hline
\end{tabular}
}
\end{table}

\section{Discussion}\label{sec:con}

We have proposed a general framework to model a directed network with communities, with network data collected by asking nodes to report or nominate their connections, a common scenario in practice.  A particular parametric form of the general model, NSBM, allows for meaningful interpretation of the parameters and computationally efficient fitting algorithms.  Other parameterizations can be set up within the general framework, perhaps for specific data collection procedures and/or applications.  The clustering algorithm we proposed applies to any such model, whereas the parameter estimation algorithm would naturally need to be derived for every parametric model separately.  In all cases, the critical point is that when we do not observe the whole network, pretending that we do tends to lead to a drop in accuracy and loss of efficiency. We saw this empirically in both simulated networks and the business school faculty hiring network.

A potential shortcoming of our current framework is that it assumes each edge is nominated independently from other edges originating from the same nominating node. This assumption will be reasonable in some applications and not so reasonable in others. Modeling dependency between edge nominations, at least those coming from the same node can be a potentially fruitful yet challenging direction for future work.  For these more complex models, estimation may have to be by variational inference or MCMC, leading to a potential trade-off on computational efficiency.  Another potentially fruitful direction is modeling network structures other than communities and investigating how incomplete and heterogeneous link nominations can affect our estimation of different types of structures and what models can be used to account and correct for the nomination process.  

\section*{Acknowledgements}

T. Li was supported in part by and NSF grant (DMS-2015298) and the Quantitative Collaborative grant from the College of Arts \& Sciences at the University of Virginia. E. Levina and T. Li (while a PhD student at the University of Michigan) were supported in part by an ONR grant (N000141612910) and NSF grants (DMS-1521551 and DMS-1916222). J. Zhu and T. Li (while a PhD student at the University of Michigan) were supported in part by NSF grants (DMS-1407698
and DMS-1821243).

\bibliography{CommonBib}{}
\bibliographystyle{abbrvnat}

\newpage

\begin{appendix}

 \section{Proofs}
\label{app:proofs} 
 
\subsection{Model identifiability}\label{appsecsec:identification-proof}

 \begin{proof}[Proof of Proposition~\ref{prop:identify}]
 We need to show that given the probability matrix $\tilde{P} = (\tilde{P}_{ij}) = (\theta_{i}B_{c_ic_j}^{\lambda_i})$ and the community labels $\V{c}$, all parameters are uniquely determined under the current constraints. Without loss of generality, we focus on identifying the parameter for one arbitrary community $k$. For any $i \in G_k$, we have
\begin{equation}\label{eq:log-decomposition}
\mu_{il}  =\log(\theta_iB_{kl}^{\lambda_i}) = \log(\theta_i) + \lambda_i\log(B_{kl})
\end{equation}
where we treat $\log(0)$ as $-\infty$. It can be seen that $\log(\theta_i) = \mu_{ik}$ by setting $l = k$ under the constraint $B_{kk}=1$.

Write $\V{b} =(\log(B_{k1}), \cdots, \log(B_{k,K}))$.  Notice that for any $1\le l \le K$ such that $B_{kl} \ne 0$, we have
\begin{equation}\label{eq:momentRelation1}
\mu_{ik} - \mu_{il} = \lambda_i(b_k - b_l).
\end{equation}
The constraint on $\lambda_i$ indicates that there exists at least one node $i$ with non-zero $\lambda_i$. Since there exists at least one $l$ such that $0 < B_{kl} \ne B_{kk}$, $b_k - b_l \ne 0$, we can locate one such node (denoted by $i_0$) and community (denoted by $l_0$) by identifying $i$ and $l$ corresponding to a non-zero $\mu_{ik} - \mu_{il}$. Given this $l_0$, we can uniquely determine the ratio between all non-zero $\lambda_i$'s. The nodes with $\lambda_i = 0$ can be directly identified from $\mu_{ik} - \mu_{il_0} =0$. Therefore, with the constraint $\sum_{i \in G_k}\lambda_k = n_k$, the identification of $\lambda_i$'s is guaranteed.

Fixing the node $i_0$, $b$ can be determined by \eqref{eq:momentRelation1} up to a shift. Since we constrain $b_k = B_{kk}=1$, all the other entries of $B_{k\cdot}$ are also identifiable. 

\end{proof}

  \subsection{Community detection}\label{appsecsec:community-proof}
 
\begin{proof}[Proof of Proposition~\ref{lem:eigen}]
It is easy to check that $\tilde{P} = FZ^T$ where $F$ is the matrix obtained by applying function $F_i$ to each element of the $i$th row of the matrix $ZB$. Write $\Delta = \diag(\sqrt{n_1}, \cdots, \sqrt{n_K})$. Assume that the SVD of $F\Delta$ is given by
$$F\Delta = UDV^T.$$
We have
$$\tilde{P} = UDV^T(Z\Delta^{-1})^T = UD(Z\Delta^{-1}V)^T.$$
Notice that $Z\Delta^{-1}$ is an orthonormal matrix and so is $Z\Delta^{-1}V$. Taking $X = \Delta^{-1}V$ gives the SVD of $\tilde{P}$, up to the standard invariances (sign-flipping and rotation within subspaces corresponding to equal singular values).Note that for a full rank $B$, $V$ is a $K\times K$ orthonormal matrix. The distance claim follows directly from the orthogonality of rows of $V$.
\end{proof}

  We will use the following three known results on spectral clustering.
  
  \begin{lem}[Lemma 7 of \cite{chen2014network}]\label{lemma:Frobenius}
  Let $M, \widehat{M}$ be two matrices of size $n\times n$ and $V, \widehat{V}$ be the $n\times K$ orthogonal matrices of top $K$ right singular vectors of $M$ and $\widehat{M}$. Then there exists a $K\times K$ orthogonal matrix $Q$ such that
  $$\norm{\widehat{V}Q-V}_F \le \frac{2\sqrt{2K}\norm{\widehat{M}-M}}{\sigma_K(M)}.$$
  \end{lem}
  
  The orthogonal matrix $Q$ makes no difference for subsequent developments and will be omitted.
  
  \begin{lem}[Lemma 5.3 of \cite{lei2014consistency}]\label{lemma:Kmeans}
  Let $V, \widehat{V}$ be two $n\times K$ matrices with $V$ having only $K$ distinct rows, corresponding to $K$ communities denoted by $\V{c}$. Let $\hat{\V{c}}$ be the output of a $K$-means clustering algorithm on $\widehat{V}$, with objective value no larger than $1+\epsilon$ of the global optimum \citep{kumar2004simple}. Denote the community indices corresponding to $\V{c}$ and $\hat{\V{c}}$ by $\{G_k\}$ and $\{\hat{G}_k\}$. Define $S_k = \{i: i\in G_k, \hat{c}_i \ne k\}$.
  For any $\delta$ smaller than the minimum distance between any two distinct rows of $V$, if
  $$8(2+\epsilon)\norm{\widehat{V}-V}_F^2 \le n_{\min}\delta^2$$
  where $n_{\min} = \min_k |G_k|$, then there exists a permutation of the $K$ community labels in $\hat{\V{c}}$, such that
  $$\sum_{k=1}^K|S_k| \le 8(2+\epsilon)\frac{\norm{\widehat{V}-V}_F^2}{\delta^2}.$$
   
  \end{lem}

  Another result we need is the concentration of a random (directed) graph adjacency matrix from \cite{le2017concentration}. A similar result was also obtained by \cite{lei2014consistency}.

\begin{lem}
\label{lemma:concentration}
Let $A$ be the adjacency matrix of a random graph on $n$ nodes with independent edges. Set $\e(A) = P = [p_{ij}]_{n\times n}$ and assume that $n\max_{ij}p_{ij}\le d$ for $d \ge C_0\log n$ and $C_0>0$. Then there exists a constant $C$ depending on $C_0$ such that 
$$\norm{A - P} \le C\sqrt{d}$$
with probability at least $1-n^{-1}$.
\end{lem}

With these three lemmas, we are ready to prove Theorem~\ref{thm:CommunityConsistency}.

\begin{proof}[Proof of Theorem~\ref{thm:CommunityConsistency}]

Let $\tilde{V}^*$  be the matrix of right singular vectors for $\tilde{P}$ and let $\tilde{V}$ be the right singular vectors of $\tilde{A}$.  The assumption $n\norm{\tilde{P}}_{\infty} \ge C_0 \log{n}$ implies the condition for concentration of Lemma~\ref{lemma:concentration}. From Lemma~\ref{lemma:concentration}, we have
$$\norm{\tilde{V}-\tilde{V}^*}_F \le \frac{2\sqrt{2K}}{\sigma_K(\tilde{P})}\norm{\tilde{A}-\tilde{P}} \le \frac{2C\sqrt{2K}}{\sigma_K(\tilde{P})}\sqrt{n\norm{\tilde{P}}_{\infty}} $$
with probability at least $1-n^{-1}$. 

To apply Lemma~\ref{lemma:Kmeans}, note that from Proposition~\ref{lem:eigen}, the minimum distance between distinct rows in $\tilde{V}^*$ is at least $\sqrt{\frac{2}{n_{\max}}}$. Therefore, according to Lemma~\ref{lemma:Kmeans},
\begin{align*}
\sum_k \frac{|S_k|}{n_k} & \le \frac{1}{n_{\min}}\sum_{k=1}^K|S_k| \le \frac{1}{n_{\min}}8(2+\epsilon)\frac{\norm{\tilde{V}-\tilde{V}^*}_F^2}{\frac{2}{n_{\max}}} \\
& \le 32C^2(2+\epsilon) \frac{n_{\max}Kn\norm{\tilde{P}}_{\infty}}{n_{\min}\sigma_K(\tilde{P})^2} \le \frac{32C^2(2+\epsilon)}{\kappa'}\frac{Kn\norm{\tilde{P}}_{\infty}}{\sigma_K(\tilde{P})^2}
\end{align*}
as long as the condition of Lemma~\ref{lemma:Kmeans} holds, 
$$\frac{32C^2(2+\epsilon)}{\kappa'} \frac{{Kn}\norm{\tilde{P}}_{\infty}}{\sigma_K(\tilde{P})^2} \le 1 \, , $$
which can be guaranteed by the assumptions of Theorem~\ref{thm:CommunityConsistency} when setting $C_1 = \frac{32C^2(2+\epsilon)}{\kappa'} $. This completes the proof.

\end{proof}


\begin{proof}[Proof of Corollary~\ref{coro:CommunityConsistency}]
Let $f_1$ and $f_2$ be the distributions of $\bar{\theta}_i$ and $\lambda_i$, respectively. We have $\norm{\tilde{P}}_{\infty} \le \rho_n \gamma_1\norm{B}_{\infty}$ from A3, where both $\gamma_1$ and $\norm{B}_{\infty}$ are constants. We now need a bound on $\sigma_K(\tilde{P})$.

From Lemma~\ref{lem:eigen}, it follows that $\sigma_K(\tilde{P})$ is the $K$-th singular value of $F=\rho_n M$ where
$$M = 
\begin{pmatrix}
  \bar{\theta}_1B_{c_1,1}^{\lambda_1} & \bar{\theta}_1B_{c_1,2}^{\lambda_1} & \cdots & \bar{\theta}_1B_{c_1,K}^{\lambda_1} \\
  \bar{\theta}_2B_{c_2,1}^{\lambda_2} & \bar{\theta}_2B_{c_2,2}^{\lambda_2} & \cdots & \bar{\theta}_2B_{c_2,K}^{\lambda_2} \\
  \vdots  & \vdots  & \ddots & \vdots  \\
  \bar{\theta}_nB_{c_n,1}^{\lambda_n} & \bar{\theta}_nB_{c_n,2}^{\lambda_n} & \cdots & \bar{\theta}_nB_{c_n,K}^{\lambda_n} \\
 \end{pmatrix}
$$
and $\Delta = \diag(\sqrt{n_1}, \cdots, \sqrt{n_K})$. Under A3, there are  at most $m_1m_2K$ distinct rows of $M$. Denote the matrix with these $m_1m_2K$ rows by $\tilde{M} \in \bR^{(m_1m_2K)\times K}$, and write
\begin{equation}\label{eq:F-form}
F = \rho_n \tilde{Z}\tilde{M} \, , 
\end{equation}
where $F$ is the same quantity in the proof of Proposition~\ref{lem:eigen},  $\tilde{Z} \in \bR^{n\times (m_1m_2K)}$ with exactly one $1$ in each row and zeros in the other positions. $\tilde{Z}$ gives the correspondence from each row of $M$ to the rows of $\tilde{M}$. Let $\tilde{n}_k$ be the number of times that the $k$th row of $\tilde{M}$ appears in rows in $M$, and define $\tilde{\Delta} = \diag(\sqrt{\tilde{n}_1}, \cdots, \sqrt{\tilde{n}_{m_1m_2K}})$. It is easy to check $\tilde{Z}\tilde{\Delta}^{-1}$ is an orthogonal matrix. Therefore,
\begin{equation}\label{eq:sigma-scale}
\sigma_K(\tilde{P}) = \sigma_K(\rho_n \tilde{\Delta}\tilde{M}\Delta) \ge \lambda \rho_n \min_{i,j,k}\sqrt{\tilde{n}_{ijk}}\min_{k}\sqrt{n_k} \, , 
\end{equation}
where $\lambda = \sigma_K(\tilde{M})$. 

By \ref{A2}, \ref{A3}, and Hoeffding's inequality, we have
$$\min_{i,j,k}\tilde{n}_{ijk} \ge C_2 n$$
with probability at least $1-\exp(-\gamma_2 n)$ for some constants $\gamma_2$, $C_2>0$ depending on $\kappa', K$ and $f_1$, $f_2$. Under this event, we have
$$\sigma_K(\tilde{P}) \ge \sqrt{C_2\kappa'}n\rho_n.$$
Finally, applying Theorem~\ref{thm:CommunityConsistency} directly gives 
$$\sum_k \frac{|S_k|}{n_k} \le C_1\frac{Kn\norm{\tilde{P}}_{\infty}}{\sigma_K(\tilde{P})^2} \le \frac{ C_1}{C_2\kappa'}\frac{K}{n\rho_n} $$
with probability at least $1-n^{-1} - e^{-\gamma_1 n}- e^{-\gamma_2 n} \ge 1-2n^{-1}$ for sufficiently large $n$. Setting $C' = \frac{C_1K}{C_2\kappa'}$ completes the proof. 
\end{proof}

\begin{lem}[Directed version of Corollary 3.6 in \cite{lei2019unified}]\label{lem:lei2019}
Let $\tilde{A} \in \{0,1\}^{n\times n}$ be an adjacency matrix of a directed network with independent Bernoulli entries and the expectation $\tilde{P} \in [0,1]^{n\times n}$. Assume  the rank of $\tilde{P}$ is $K$ and $K$ is fixed. Let $\tilde{A} = \hat{U}\hat{\Sigma} \hat{V}^T$ and $U \Sigma V^T$ be the rank $K$ SVD of $\tilde{A}$ and $\tilde{P}$, respectively.  If 
$$\Sigma_{KK} \ge C_0 n\norm{\tilde{P}}_{\infty},$$
and $n\norm{P}_{\infty} \ge C_0\log{n}$ for some constant $C_0 >0$, 
then with probability at least $1-n^{-1}$, we have
$$\max(\norm{\hat{U} - U}_{2,\infty}, \norm{\hat{V} - V}_{2,\infty} ) \le C\sqrt{\frac{\log{n}}{n\norm{\tilde{P}}_{\infty}}}\max\left( \norm{U}_{2,\infty}, \norm{V}_{2,\infty}     \right).$$
\end{lem}
\begin{proof}[Proof of Lemma~\ref{lem:lei2019}]
  Let
$$\tilde{P}^{s} = \begin{bmatrix}
0  & \tilde{P}\\
\tilde{P}^T & 0
\end{bmatrix}  \text{  and  }
\tilde{A}^s = 
\begin{bmatrix}
0  & \tilde{A}\\
\tilde{A}^T & 0
\end{bmatrix}.
$$
Then $\tilde{A}^s$ is a symmetric matrix with independent Bernoulli entries in the upper triangular positions, drawn with probabilities $\tilde{P}^s$.  The eigenvectors of $\tilde{P}^s$ are $[U^T, V^T]^T$.  The result follows directly by applying Corollary 3.6 of \cite{lei2019unified} with the additional constraint on $\Sigma_{KK}$, which corresponds to formula (200) of \cite{lei2019unified}. 
\end{proof}

 \begin{proof}[Proof of Theorem~\ref{thm:CommunityConsistency-mst}]
By Proposition~\ref{lem:eigen} and \ref{A2}, the minimum spanning tree alrogithm will perfectly recover communities if 
 \begin{equation}\label{eq:distance-separation}
 \max_i\norm{\hat{V}_{i\cdot} - V_{i\cdot}} < \frac{\sqrt{2}}{4\sqrt{\left (1-\left(K-1\right)\kappa'\right)n}} \le  \frac{\sqrt{2}}{4\sqrt{n_{\max}}} \le \frac{1}{4}\min_{c_i \ne c_j} \norm{V_{i\cdot} - V_{j\cdot}}.
 \end{equation}
 This is because \eqref{eq:distance-separation} ensures that any between-community edge would have a higher weight than any within-community edge.  In the minimum spanning tree, between any two communities, there is at most one edge connecting them, and in total there would be exactly $K-1$ between-community edges.  Therefore, removing the $K-1$ edges with largest weights results in the correct community partition.

 It remains to show \eqref{eq:distance-separation}.  Lemma~\ref{lem:lei2019} gives
 $$ \max_i\norm{\hat{V}_{i\cdot} - V_{i\cdot}}  = \norm{\hat{V}-V}_{2,\infty} \le C\sqrt{\frac{\log{n}}{n\norm{\tilde{P}}_{\infty}}}\max\left( \norm{U}_{2,\infty}, \norm{V}_{2,\infty}     \right)$$
 with probability at least $1 - n^{-1}$, which implies \eqref{eq:distance-separation}.
 
 \end{proof}

\begin{proof}[Proof of Corollary~\ref{coro:CommunityConsistency-mst-NSBM}]
To apply Theorem~\ref{thm:CommunityConsistency-mst}, we just need to show the two conditions \eqref{eq:generic-basic-signal} and \eqref{eq:generic-strong-signal} hold.    Equation \eqref{eq:sigma-scale} in the proof of Corollary~\ref{coro:CommunityConsistency} implies \eqref{eq:generic-basic-signal}.   As discussed after Corollary~\ref{coro:CommunityConsistency-mst-NSBM}, to show $\eqref{eq:generic-strong-signal}$ when $n\rho_n/\log{n} \to \infty$, it is sufficient to show that $\tilde{P}$ is perfectly incoherent, that is 
$$\norm{U}_{2,\infty} = O(1/\sqrt{n}) \text{~~~and~~~} \norm{V}_{2,\infty} = O(1/\sqrt{n}).$$

From Proposition~\ref{lem:eigen}, we know that $V$ has only $K$ distinct rows and each unique row appears at least $n_{\min}$ times, from \ref{A2}. Therefore, $\norm{V}{2,\infty} = O(1/\sqrt{n})$. 

The proof of Proposition~\ref{lem:eigen} indicates that $U$ consists of the left singular vectors of $F$, which is given by \eqref{eq:F-form}. Using the same notation, we have
$$F = \tilde{Z}\tilde{\Delta}^{-1}\Delta \tilde{M} \, , $$
where $\Delta \tilde{M}$ is an $m_1m_2K\times K$ matrix. Again, $\tilde{Z}\tilde{\Delta}^{-1}$ is an orthonormal matrix. Therefore, $U = \tilde{Z}\tilde{\Delta}^{-1}\tilde{U}$ where $\tilde{U}$ is the left singular vector of $\Delta \tilde{M}$. Hence $U$ only has $m_1m_2K$ distinct rows. Since we assume $m_1$,  $m_2$, and $K$ to be fixed, we also know $\norm{U}_{2,\infty} = O(1/\sqrt{n})$.

\end{proof}

\subsection{Proofs for parameter estimation under the NSBM}\label{appsecsec:PE-proof}

 \begin{proof}[Proof of Theorem~\ref{thm:IdeaEstimationErrorBound}]
Without loss of generality, let us assume  the first $n_1$ nodes are from community 1 and focus on estimating parameters in community $1$. The same argument can be repeated for the other $K-1$ communities. Note that consistency trivially holds for $B_{1l}=0$, so for this proof we focus on the case $B_{1l}>0$. For each $l\in [K]$ such that $B_{1l}>0$, define
$$\tilde{P}_{il} = \theta_iB_{1l}^{\lambda_i}.$$
By Bernstein inequality, we have
\begin{equation}\label{eq:Bernstein-1}
\p\left( \left|\frac{\sum_{j\in G_{l}}A_{ij}}{n_{l}} - \tilde{P}_{il}\right| > t \right) \le 2\exp \left(-\frac{n_lt^2/2}{\tilde{P}_{il} + t/3}\right).
\end{equation}
To make the concentration nontrivial, we need to require at least $t \le \tilde{P}_{il}$. Hence we have $\tilde{P}_{il} \ge t/3$, leading to
\begin{equation}\label{eq:basicConcentration}
  \p  \left( \left|\frac{\sum_{j\in G_{l}}A_{ij}}{n_{l}} - \tilde{P}_{il}\right| > t\right)
  \le 2\exp\left(-\frac{n_lt^2}{4\tilde{P}_{il}}\right) \le 2\exp\left(-\frac{\kappa' nt^2}{4\tilde{P}_{il}}\right).
\end{equation}

A useful special case is $t = \delta_n \tilde{P}_{il}$, for which \eqref{eq:basicConcentration} gives 
\begin{equation}\label{eq:RelativeConcentration}
\p(|T_{il} - \tilde{P}_{il}| > \delta_n \tilde{P}_{il} ) \le 2\exp\left(-\frac{\kappa' }{4}n\delta_n^2\tilde{P}_{il}\right).
\end{equation}
When $l=1$, we have $\tilde{P}_{i1} = \theta_i$.  If $\theta_i \ge \frac{8}{\kappa'}\frac{\log^4{n}}{n}$, setting $t = \theta_i/\log{n}$ in \eqref{eq:RelativeConcentration} gives
\begin{equation}\label{eq:thetai}
\p\left(|\hat{\theta}_i - \theta_i|/\theta_i > \frac{1}{\log{n}} \right)\le  2\exp\left(-\frac{\kappa' n\theta_i}{4\log^2{n}}\right) \le 2\exp(-2\log^2{n}).
\end{equation}
Therefore, when $\min_i \theta_i \ge  \frac{8}{\kappa'}\frac{\log^4{n}}{n}$,  taking the union over all $i \in [n]$ gives 
$$\p(\max_i|\hat{\theta}_i - \theta_i|/\theta_i > 1/\log{n})\le 2n\exp(-2\log^2{n}) \le \exp(-\log^2{n}) \le n^{-1}.$$
for sufficiently large $n$.  This finishes the proof of Part 1.

Recall that we only need to consider $l$ such that $B_{1l > 0}$. Using the same  \eqref{eq:RelativeConcentration}, we can also see that  $T_{il} > 0$ with high probability, so we will treat $\log{T_{il}}$ as well-defined from now on.

Let $\mu_{il} = \log(\tilde{P}_{il})$. We now proceed to boundl $Y_{il} - \mu_{il}  = \log(T_{il}) -  \log(\tilde{P}_{il})$. A useful inequality is, for $x$, $y > 0$, 
 $$|\log(x)-\log(y)| \le \frac{|x-y|}{\min(x,y)} \, . $$
We have
 \begin{align*}
 \p(|Y_{il} - \mu_{il}| > t) &=  \p(|Y_{il} - \mu_{il}| > t, \; |T_{il}-\tilde{P_{il}}| \le \tilde{P}_{il}/2) +  \p(|Y_{il} - \mu_{il}| > t, |T_{il}-\tilde{P_{il}}| > \tilde{P}_{il}/2)\\
 & \le \p(|Y_{il} - \mu_{il}| > t, |T_{il}-\tilde{P_{il}}| \le \tilde{P}_{il}/2) +  \p( |T_{il}-\tilde{P_{il}}| > \tilde{P}_{il}/2)\\
 & \le \p\left(    \frac{|T_{il} - \tilde{P}_{il} |}{\min(T_{il} ,\tilde{P}_{il} )} >t   , |T_{il}-\tilde{P_{il}}| \le \tilde{P}_{il}/2\right) +  \p( |T_{il}-\tilde{P_{il}}| > \tilde{P}_{il}/2)\\
 & \le \p\left(    \frac{|T_{il} - \tilde{P}_{il} |}{\tilde{P}_{il}/2} >t   , |T_{il}-\tilde{P_{il}}| \le \tilde{P}_{il}/2\right) +  \p( |T_{il}-\tilde{P_{il}}| > \tilde{P}_{il}/2)\\
 & \le \p(|T_{il} - \tilde{P}_{il}|/\tilde{P}_{il} > t/2, |T_{il}-\tilde{P_{il}}| \le \tilde{P}_{il}/2) +  \p( |T_{il}-\tilde{P_{il}}| > \tilde{P}_{il}/2)\\
 & \le \p(|T_{il} - \tilde{P}_{il}|/\tilde{P}_{il} > t/2) +  \p( |T_{il}-\tilde{P_{il}}| > \tilde{P}_{il}/2) \, . \\
 \end{align*}
By setting $\delta_n = t/2$ and $1/2$ in \eqref{eq:RelativeConcentration}, we get 
\begin{equation}\label{eq:ConcentrationLog-general}
\p(|Y_{il} - \mu_{il}| > t)  \le 2\exp(-\frac{\kappa'}{16}nt^2\tilde{P}_{il}) + 2\exp(-\frac{\kappa'}{16}n\tilde{P}_{il}).
\end{equation}

Moreover, combining \eqref{eq:ConcentrationLog-general} and its version when $l = 1$ and constraining $t < 1$, we have
\begin{equation}\label{eq:denominator-general}
\p(|(Y_{i1}-Y_{il}) - (\mu_{i1}-\mu_{il})| > t)  \le 4\exp(-\frac{\kappa'}{16}nt^2\tilde{P}_{il}) + 4\exp(-\frac{\kappa'}{16}n\tilde{P}_{il}) \le 8\exp(-\frac{\kappa'}{16}nt^2\tilde{P}_{il}).
\end{equation}
In particular, under \ref{A3}, there is a constant  $\phi \le \min_{ij}\bar{\theta}_iB_{c_ic_j}^{\lambda_i}$.  Now let $\eta = 2\max_{ij}B_{ij}$, $t = \frac{1}{\eta \log{n}}$. If we assume 
$$\rho_n \ge \frac{16\eta^2}{\phi\kappa'}\frac{\log^4{n}}{n} \, , $$ 
then 
\begin{equation}\label{eq:denominator-logB}
\p(|(Y_{i1}-Y_{il}) - (\mu_{i1}-\mu_{il})| >  \frac{1}{\eta \log{n}})  \le 8\exp(-\frac{\kappa'}{16\eta^2}\frac{n}{\log^2{n}}\phi\rho_n) \le 8\exp(-\log^2{n}).
\end{equation}

Recall that we have $\mu_{i1}-\mu_{il} = -\lambda_i \log(B_{1l})$. With the constraint $\sum_{i \in G_1}\lambda_i = n_1$, we have $\sum_{i \in G_1}(\mu_{i1}-\mu_{il})/n_1 = -\log(B_{1l}) $. Therefore, applying \eqref{eq:denominator-logB} to all nodes in $G_1$, we have
\begin{align}\label{eq:numerator-logB}
  \p \left(|\log{\hat{B}_{1l}} - \log{B_{1l}}| > \frac{1}{\eta \log{n}}\right) &=
\p \left( \left|\frac{1}{n_1}\sum_{i: \V{c}_i = 1}[(Y_{i1}-Y_{il}) - (\mu_{i1}-\mu_{il})]\right|> \frac{1}{\eta \log{n}} \right) \notag \\
& \le 8n_1\exp(-\log^2{n}).
\end{align}
Because the function $\exp(x)$ is convex, for any $x, y>0$ we have
$$|\exp(x)-\exp(y)| \le |x-y|\exp(\max(x,y)).$$

For sufficiently large $n$, $\frac{1}{\eta \log{n}} < \log 2$. Under the event of \eqref{eq:numerator-logB},
$$|\hat{B}_{kl} - B_{kl}| \le  \exp(\max(\log{B_{kl}}, \log{\hat{B}_{kl}})) \frac{1}{\eta \log{n}} \le \exp(\log{B_{kl}}+\log{2})\frac{1}{\eta \log{n}} \le \eta \frac{1}{\eta \log{n}} \le \frac{1}{\log{n}}.$$
Part 2 of the theorem comes directly from \eqref{eq:numerator-logB} after taking the union of at most $K^2$ events for community pairs with nonzero $B_{kl}$. The event probability is then controlled by $1- 8Kn\exp(-\log^2{n}) \ge 1 - n^{-1}$ for sufficiently large $n$.

This argument can be improved by leveraging independence between different rows of $\tilde{A}$, instead of directly taking the union of events across $i$. That approach can slightly reduce the requirement on $\rho_n$.

For Part 3, first note that because of the previous discussion, we only consider the settings when $T_{il} > 0$ for $B_{1l} > 0$. Therefore, we treat $\Psi_1$ as known.  

First consider the true parameter values.  For any $l \in \Psi_1$, define $b_{l} = \log(B_{1l})$ for $B_{1l}>0$. We have
$$\mu_{il}-\mu_{i1} = \log{\tilde{P}_{il}} - \log{\tilde{P}_{i1}} = \lambda_i (b_l-b_1), i \in G_1.$$
Summing up across $\psi_1$, 
$$\sum_{l \in \Psi_1} (\mu_{il}-\mu_{i1} ) = \lambda_i \sum_{l \in \Psi_1}(b_l-b_1).$$

Under the identifiability constraint, we also have
$$\frac{1}{n_1}\sum_{i \in G_1}\sum_{l \in \Psi_1} (\mu_{il}-\mu_{i1}) = \sum_{l \in \Psi_1}(b_l-b_1).$$
The two identities give
$$\frac{\sum_{l \in \Psi_1} (\mu_{il}-\mu_{i1})}{\frac{1}{n_1}\sum_{i \in G_1}\sum_{l \in \Psi_1} (\mu_{il}-\mu_{i1}) } = \lambda_i.$$

To obtain an error bound for estimated parameters, we will separately bound  the numerator and the denominator above.    First, apply \eqref{eq:denominator-general} for all $l \in \Psi_1$, obtaining 
\begin{equation}\label{eq:denominator}
\p(|\sum_{l \in \Psi_1}(Y_{i1}-Y_{il}) - \sum_{l \in \Psi_1}(\mu_{i1}-\mu_{il})| > Kt)  \le 8K\exp(-\frac{\kappa'}{4}nt^2\tilde{P}_{il}).
\end{equation}
Then applying \eqref{eq:denominator} across $i\in G_1$, we have
\begin{equation}\label{eq:numerator}
\p\left(\left|\frac{\sum_{i\in G_1}\sum_{l \in \Psi_1}(Y_{i1}-Y_{il})}{n_1} - \frac{\sum_{i\in G_1}\sum_{l \in \Psi_1}(\mu_{i1}-\mu_{il})}{n_1})\right| > Kt\right)  \le 8n_1K\exp(-\frac{\kappa'}{4}nt^2\tilde{P}_{il}).
\end{equation}
Another useful inequality we need, for any $x, y, x_0, y_0$ such that $x\cdot x_0 > 0, y \cdot y_0 > 0$, is 
 \begin{equation}\label{eq:ratiobound}
   \left|\frac{x}{y} - \frac{x_0}{y_0}\right|
   \le
   \sqrt{\frac{1}{\min(|y|, |y_0|)^2}+\frac{\max(|x|,|x_0|)^2}{\min(|y|, |y_0|)^4}}\left(|x-x_0| + |y-y_0|  \right).
 \end{equation}
 
By \ref{A3}, there are constants $\alpha, \beta > 0$ such that 
$$|\lambda_i| < \alpha\text{~~~and~~~} 1/\beta < \min_k\sum_{l \in \Psi_k}(b_l-b_1) \le \max_k \sum_{l \in \Psi_k}(b_l-b_1) < \beta.$$
Under the complement event of the union of \eqref{eq:denominator} and \eqref{eq:numerator} and assuming $Kt < \beta/2$, we apply \eqref{eq:ratiobound} with 
$$x_0 = \sum_{l \in \Psi_1}(\mu_{i1}-\mu_{il}), y_0 = \frac{\sum_{i\in G_1}\sum_{l \in \Psi_1}(\mu_{i1}-\mu_{il})}{n_1} \, , $$
$$x =\sum_{l \in \Psi_1}(Y_{i1}-Y_{il}) , y = \frac{\sum_{i\in G_1}\sum_{l \in \Psi_1}(Y_{i1}-Y_{il})}{n_1}.$$
This gives
\begin{align*}
|\hat{\lambda}_i-\lambda_i| \le  \sqrt{\frac{1}{(|y_0|/2)^2}+\frac{(|x_0|+Kt)^2}{ (|y_0|/2)^4}}2Kt \le 2K\sqrt{4\beta^2+16\beta^4(\alpha+3\beta/2)^2} t
\end{align*}

Now set $t = \frac{1}{2K\sqrt{4\beta^2+16\beta^4(\alpha+3\beta/2)^2}}\frac{1}{\log{n}}$ in \eqref{eq:denominator} and \eqref{eq:numerator}. It is not difficult to see that as long as
$$\rho_n \ge \frac{16K^2(4\beta^2+16\beta^4(\alpha+3\beta/2)^2)}{\phi\kappa'}\frac{\log^4{n}}{n},$$ 
we have
\begin{equation}\label{eq:lambda-bound}
\p(|\hat{\lambda}_i-\lambda_i| > \frac{1}{\log{n}}) \le 8(n_1+1)\exp(-\log^2{n}).
\end{equation}

Again, repeating this argument for all $i \in [n]$ leads to
\begin{equation}\label{eq:lambda-final}
\p(\max_i |\hat{\lambda}_i-\lambda_i| \le \frac{1}{\log{n}}) \ge 1 - 16Kn^2\exp(-\log^2{n}) \ge 1-n^{-1}
\end{equation}
for sufficiently large $n$. This completes the proof of Part 3.

 \end{proof}

\end{appendix}

\end{document}